\documentclass[a4paper,11pt]{article}
\usepackage{amsmath,amssymb}
\usepackage[dvips]{graphicx}
\usepackage{subcaption}
\usepackage{color}
\usepackage{authblk}
\usepackage{bbm}
\usepackage{hyperref}
\usepackage{dsfont}

\newenvironment{proof}{{\bf Proof.}}{\hfill\mbox{$\Box$} \\ }

\newtheorem{prop}{Proposition}[section]

\newcommand{\bE}{{\Bbb E}}

\newcommand{\bQ}{{\Bbb Q}}
\newcommand{\bN}{{\Bbb N}}
\newcommand{\bR}{{\Bbb R}}

\def\1{{\bf 1}}

\def\epsilon{\varepsilon}

\DeclareMathOperator{\arcsinh}{arcsinh}

\title{Capturing the power options smile by an additive two-factor
model for overlapping futures prices}

\author{Marco Piccirilli}
\affil[1]{Dublin City University, School of Mathematical Sciences, Glasnevin, Dublin 9, Ireland}

\author{Maren Diane Schmeck} 
\affil[2]{Center for Mathematical Economics, University of Bielefeld, Universit\"atssta\ss e 25, D-33615 Bielefeld, Germany} 

\author{Tiziano Vargiolu}
\affil[3]{University of Padova, Department of Mathematics, via Trieste 63, Torre Archimede, I-35121 Padova, Italy}

\begin{document}

\maketitle

\begin{abstract} 
In this paper we introduce an additive two-factor model for  electricity futures prices based on Normal Inverse Gaussian L\'evy processes, that fulfills a no-overlapping-arbitrage (NOA) condition. We compute European option prices by Fourier transform methods, introduce a specific calibration procedure that takes into account no-arbitrage constraints and fit the model to  power option settlement prices of the European Energy Exchange (EEX). We show that our model is able to reproduce the different levels and shapes of the implied volatility (IV) profiles displayed by options with a variety of delivery periods.
\end{abstract}



\section{Introduction}

One of the big challenges of this age is to develop efficient storage possibilities for electricity. As the available storage possibilities are very limited in efficiency and capacity, a major breakthrough of technology is necessary to make the storability of electricity comparable to those of  commodities as corn or oil.  Thus, at this point of time one can say (simplifying) that  electricity is  not storable, at least not in the sense of classical commodities. This has important consequences. One of them is that electricity futures contracts deliver the underlying not at a point of time, but rather over a specified delivery period of different length: typically one can find monthly, quarterly and yearly delivery periods. Obviously, one quarter consist of three months and one year of four quarters, such that there might be futures with overlapping delivery periods offered in the market. Compared to other markets, this gives rise to an additional dimension of arbitrage opportunities, which appears when trading in futures with overlapping delivery periods. Recently, \cite{BPV} provide new theoretical results concerning the no-arbitrage theory of overlapping additive futures models. In this framework, the present paper addresses the arbitrage free pricing and calibration of options on overlapping electricity  futures.\\

Trading in electricity derivatives is living a significant expansion in Europe. In June 2018, the European Energy Exchange increased volumes on its power derivatives markets by 28\% from 181.2 TWh (June 2017) to 231.1 TWh \cite{eex2}. In particular, the trading volume in power options experienced a boost of 45\%. Additionally, these contracts are traded over the counter. Thus, it is not a surprise that also the literature considering electricity options as well as  options on other commodities is booming. One stream of literature   considers  e.g. seasonality, stochastic volatility and Samuelson-like effects in pricing of power options (e.g. \cite{schmeck76,schmeck_rp,  fanelli,nastasi,nomikos,schmeck2016pricing,wottka}) or  commodity options (e.g. \cite{arismendi,back,schneider16,schneider18}). However, the issue of overlapping delivery periods is not addressed in the above mentioned articles. \\

When building no-arbitrage models with overlapping deliver periods (from now on called ``no-overlapping-arbitrage''  or shorter NOA) the focus must be twofold. First, one has to work in a risk neutral setting where traded contracts are martingales. When such a risk neutral measure exists, the first fundamental theorem of asset pricing guarantees that there are no arbitrage possibilities. In fact, when considering overlapping periods,  one must additionally take into account possible arbitrages arising from trading futures with different delivery periods (cf. \cite{latini}). The first paper that takes into account this fact is \cite{kieselSB}, where the authors apply an approach inspired by LIBOR market models \cite{brigo_mercurio}. To the best of our knowledge, this is also the first attempt to fit a consistent option pricing model for power markets. However, the generic LIBOR approach used in \cite{brigo_mercurio} allows to model directly only shortest delivery, i.e. monthly, contracts, while resorting to approximation to the distribution of longest delivery futures. The two-factor model of \cite{kieselSB} has been generalized by \cite[Section 5]{borger} to L\'evy models, allowing to describe the implied volatility (IV) surface of option prices. Furthermore, a more general approximation procedure is proposed there, but the results, though accurate for single deliveries, are partially satisfactory for longest maturity contracts, probably due to the required approximations. Here, we see a strength of our additive approach,  making the necessary approximations of the geometric approach useless.\\

Among recent contributions in power options modeling, we mention \cite{wottka}, who uses a structural model that explains the formation of IV skews, and \cite{nomikos}, who develop an extensive sensitivity analysis of IV patterns. However, differently from our setting, both papers consider options written on electricity spot prices. Another characteristic that we want to reproduce in our model is prices' seasonality (see also \cite{borovkova2017electricity}). Here, we address seasonality in the sense of dependence on the delivery of the underlying futures. This has been addressed by \cite{fanelli}, though in a Gaussian setting, which does not allow to capture the different IV levels displayed by options with different strikes. \\


In order to compute option prices, we introduce an arbitrage-free two-factor additive model for futures prices. In contrast to spot-based models, where the futures prices are computed by taking the discounted expectation of spot prices under a pricing measure, we start by specifying the stochastic dynamics of futures prices under a risk-neutral measure $\bQ$. This is known to literature as the Heath-Jarrow-Morton approach applied to commodity markets (see e.g. \cite{musti} for a literature review in the case of electricity). Since the parameters of our model will be estimated from the option market, this approach has several advantages as already pointed out e.g. by \cite{kieselSB}. For instance, we do not need to choose a pricing measure, or, equivalently, estimate the market price of risk that determines how the stochastic model changes from the physical measure to the risk-neutral measure under which the options are priced in the market. In fact, we can call our risk neutral measure $\bQ$, \emph{the} risk-neutral measure, since this is implicitly and univocally determined by the option prices observed in the market.
We refer to \cite{schmeck_rp} for an empirical  discussion on pricing measures for electricity derivatives. 

Since  our futures price model  is based on L\'evy processes, it allows to capture the implied volatility profile described by options with different strikes. 
The dynamics arises as the natural generalization of its Gaussian counterpart introduced in \cite{latini}. It assumes that futures prices are stirred by two stochastic factors built on Normal Inverse Gaussian (NIG) L\'evy processes modulated by deterministic coefficients depending both on time and delivery period. 
The NIG distribution is a flexible family of distributions that is very popular in financial modeling (see e.g. \cite{BN97} for general applications of NIG processes to finance, \cite{benth04} for an electricity spot price NIG model and \cite{andresen} for modeling electricity forwards). 
The first factor has a delivery-averaged exponential behavior, meant to reproduce the Samuelson effect \cite{lautier}. The second factor is independent of current time, but varies for contracts with different delivery in a \emph{seasonal} and \emph{no-arbitrage} way (cf. \cite{latini}) and accounts for a finer reproduction of the term structure of futures' volatilities.
We use an \emph{additive} model, meaning that we do not consider the log-prices, but we instead model directly the prices. This class of models, as opposed to geometric models, has recently gained an increasing attention in literature due to many modeling advantages (see, for example, \cite{MR2323278,kufa,BPV,gallana,kieselpara}). 
In the context of option pricing, \cite{kufa} exploit the additive structure of their spot dynamics for pricing Asian and spread options by fast Fourier techniques. In our case, the additivity property will allow us to fit our model consistently to all the delivery periods traded in the market. Instead, the calibration results of \cite[Section 5]{borger}, where the author studied a geometric version of our model (still based on two NIG factors) were  only partially satisfactory. As already mentioned at the beginning, this is likely due to the necessity to introduce an approximation procedure for the distribution of contracts covering periods of more than one month. In our case instead, by considering an additive model, we do not need to introduce such approximations, as we have an exact expression for the contract's distribution. 
For pricing options, we follow the classical method of \cite{CM}, which consists, roughly speaking, in computing the Fourier transform of (a suitable modification of) the call payoff as a function of the strike price so to recover the option value by inverse transform. 
We recall from \cite{CM} two different approaches, that we adapt to the case of additive models. The first approach is based on the use of an exponential damping factor in order to make the option value integrable on the whole real line. Instead, the second approach consists in substracting the time value of the option from its payoff. In this way we derive semi-analytical expressions (analytical up to numerical integration) for the option prices, that depend on the characteristic function of the underlying (see also \cite{MR2416686}). 

We discuss a calibration methodology that we will apply in our empirical study.  
The calibration happens statically, in the sense that we fix a trading date and observe the market option prices for different strikes and deliveries. 
We then find the parameters that minimize an objective function representing the distance from  observed prices to model prices. In the same way, one can alternatively use IVs instead of option prices. The futures model under consideration is defined in such a way that there is no possibility of arbitrages from trading in overlapping delivery periods. Because no-arbitrage implies certain relations on the coefficients, this  translates into parameter constraints (cf. the calibration procedure presented in \cite{latini}). 

Since there is not enough liquidity in the market in order to extract information on the IV surface from traded market quotes, we consider the options settlement prices, that are available for a sufficiently large range of strike prices. Though settlement prices do not necessarily represent trades that take place in the market, they contain information on market expectations. We perform the calibration procedure described above, first, for a one-factor model derived by the two-factor one by setting one coefficient to 0, and then for the general two-factor model. We compare the IVs of both models to the empirical IVs and the one (constant across strikes) generated by Black's model \cite{schmeck76,black}. 
As a by-product from the estimation of the one-factor model, we derive that, under the risk-neutral measure, futures prices are leptokurtic and have significantly positive skewness. This is reflected also into the shape of empirical IVs, which display a \emph{forward skew} (i.e. higher IVs for out-of-the-money calls). This can be interpreted as a ``risk premium'' paid by option buyers for securing supply (cf. \cite{birkelund,taylor11}). 
Finally, we show that the two-factor model is able to reproduce in a satisfactory way the different levels and shapes of the IV profiles displayed by all the deliveries traded in the market, by outperforming both the Black and the one-factor model.  \\

This paper is organized as follows. In the forthcoming Section \ref{sec:2} we state our modeling framework and discuss the no-overlapping-arbitrage  conditions  in this setting. Building on this,  we address options on futures as well as its pricing by fast Fourier methods in the following section. Then, Section \ref{sec:4} is devoted to the calibration procedure when considering overlapping delivery periods and Section \ref{sec:5} to the empirical study. Finally, Section \ref{sec:6} concludes.

\section{Additive multifactor models for futures contracts in an overlapping arbitrage free framework}\label{sec:2}

Let $F(t,T_1,T_n)$ denote the price at a given day $t$ of a futures contract which delivers a fixed intensity of electricity over the period $[T_1,T_n]$. This period is divided into subperiods $[T_i, T_{i+1}]$, for $i=1, \cdots, n-1$. Thus, the periods are overlapping. For example, the delivery period $[T_1,T_n]$ can be one quarter long, and is divided into four monthly periods, or one year and it is divided into four quarters. Since contracts expire right before delivery starts, we have that $t\leq T_1$.  We will state the multifactor model in this framework. Nevertheless, when it is sufficient to consider one period only, we will use the period $[T_1, T_2]$ representative for a single arbitrary period.
We introduce a general framework for multifactor additive futures prices (as, for example, in \cite{MR2416686}) and, from this, we focus on a two-factor model inspired by \cite{latini}, that will be of interest for application. More in detail, we introduce a stochastic evolution, parametrized by the delivery period (i.e. depending, in addition to the trading day $t$, also on $T_1,T_2$), driven by independent L\'evy factors. We also compute the corresponding characteristic functions, that will constitute the main ingredients in the computation of option prices (see Section 4.3). We will assume throughout this work that the risk-free interest rate is zero (similar arguments apply in the case of deterministic flat interest rate after discounting, see \cite{latini}).  


\subsection{Additive multifactor NOA models}

In the rest of the paper, we will consider European options written on futures contracts written on the delivery period $[T_1, T_2]$ and we will denote by $T<T_1$ the exercise date of these options. Therefore, for convenience we express the futures price at time $T$. 
We assume that, for any time $t$ before the exercise of the option, i.e. $0\leq t<T$, the futures prices $F(T,T_1,T_2)$ are given by 
the following stochastic differential equation (here given in integral form) 
\begin{align}\label{multifactor}
F(T,T_1,T_2)&=F(t,T_1,T_2)+\sum_{k=1}^p \int_t^T \Sigma_k(u,T_1,T_2)\,dW_k(u)
+\sum_{j=1}^m \int_t^T \Gamma_j(u,T_1,T_2)\,dJ_j(u)
\end{align}
where $W_k$ are independent Brownian motions for $k=1,\ldots,p$ and $J_j(u)=\int_0^u\int_\bR y\,\widetilde N_j(dy,dv)$ are independent, pure-jump, centered L\'{e}vy processes such that $\int_{|y|>1} y^2\,\nu_j(dy)<\infty$ for each $j=1,\ldots,m$. This integrability assumption on the L\'evy measure implies that $J_j$ are square-integrable martingales with zero expectation (for background on L\'evy processes see e.g. \cite{CT}). We assume that all the stochastic factors are independent, so that, in particular, the $m$ Poisson random measures are independent of the $p$ Brownian components. The dynamics above are described under a risk-neutral measure $\bQ$. The absence of the drift follows from the fact that, by no-arbitrage, futures prices must be martingales under $\bQ$ (see e.g. \cite{MR2057475}). 

It is often the case, in power markets, that futures written on overlapping periods are traded simultaneously: this originates from the so-called {\em cascade mechanism},  by which calendar and quarterly futures are split in contracts spanning smaller periods, see e.g. \cite{BPV} or \cite{latini} for a graphical illustration. In particular, it may happen that, for a period $[T_1,T_n]$, the futures $F(t,T_1,T_n)$ and $(F(t,T_i,T_{i+1}))_{i=1,\ldots,n-1}$ are all quoted at the same time $t < T_1$. In this case, a naturally arising condition on the value $V(t, K ;T_1, T_n)$ of a futures contract with price $K$ over the period $[T_1, T_n]$ and the value $V(t,K;T_i,T_{i+1})$ with the same price $K$ is
\begin{align}\label{eq:arbV}
V(t, K ;T_1, T_n)= \sum_{i=1}^{n-1}V(t,K;T_i,T_{i+1})\;.
\end{align}
That is, the value of receiving electricity for the fixed price $K$ over the period $[T_1, T_n]$ has to be the same as receiving the electricity for the same price over all partial periods.  Now, the forward price $K=F(t, T_1, T_n)$ is defined to be the price that makes the value $V(t, K ;T_1, T_n)$ of the contract being equal to zero. Thus, from   \eqref{eq:arbV} it follows the following no-overlapping-arbitrage (NOA) condition on the forward prices:
\begin{align}\label{multifactori}
F(t,T_1, T_n)= \frac{1}{T_n-T_1}\sum_{i=1}^{n-1}(T_{i+1}-T_i)F(t,T_i,T_{i+1}).
\end{align}
For more of this topic, see \cite{latini} or \cite{borger}. Considering the dynamics of \eqref{multifactor} and \eqref{multifactori}, the NOA condition on the futures price reads in our framework 

\begin{align}
\Sigma_k(t,T_1,T_n) &= \frac{1}{T_n-T_1}\sum_{i=1}^{n-1}(T_{i+1}-T_i)  \Sigma_k(t, T_i,T_{i+1}) \label{eq:NOA1}\\   
\Gamma_j(t,T_1,T_{n}) &= \frac{1}{T_n-T_1}\sum_{i=1}^{n-1}(T_{i+1}-T_i) \Gamma_j( t
,T_i,T_{i+1}) \label{eq:NOA2}
\end{align}
for all $t < T_1$, $k=1,\dots,p$ and $j=1,\dots, m$, see e.g. \cite{BPV}.

In order to apply the Fourier transform approach to option pricing, we need the characteristic function of the underlying process. Consider now a generic period $[T_1, T_2]$.
By introducing $Z(t,T,T_1,T_2) := F(T,T_1,T_2)-F(t,T_1,T_2)$ and its characteristic function (as a function of $v\in\bR$)
$$
\Psi(t,T,T_1,T_2,v) = \bE\left[ e^{ivZ(t,T,T_1,T_2)} \left|\mathcal{F}_t\right.\right],
$$
we have that (see e.g. \cite{MR2416686})
\begin{align}\label{cumulants}
\log\Psi(t,T,T_1,T_2,v) &= -\frac12 v^2 \sum_{k=1}^p \int_t^T \Sigma^2_k(u,T_1,T_2)\,du +\sum_{j=1}^n \psi_j(t,T;v\Gamma_j(\cdot,T_1,T_2)).
\end{align}
The function $\psi_j(t,T;\theta(\cdot))$ denotes 
$$
\psi_j(t,T;\theta(\cdot))=\int_t^T \widetilde\psi_j(\theta(u))\,du=\int_t^T \int_\bR (e^{i\theta(u) z}-1-i \theta(u) z )\,\nu_j(dz)\,du,
$$
where $\widetilde\psi_j(\theta)$ is the cumulant of the L\'{e}vy process $J_j$ computed in $\theta\in\bR$, i.e. 
$$
\widetilde\psi_j(\theta)=\log\bE\left[e^{i\theta J_j(1)}\right]
$$
and $\nu_j$ is the L\'{e}vy measure of $J_j$.

\subsection{A two-factor model based on Normal Inverse Gaussian processes}

In this section we consider a two-factor model of the type \eqref{multifactor} based on the Normal Inverse Gaussian (NIG) distribution. This is motivated by the framework \cite{BPV} and arises as a natural generalization of \cite{latini}. The model assumes that futures prices are stirred by two stochastic factors built on Normal Inverse Gaussian L\'evy processes modulated by deterministic coefficients. The first factor has a delivery-averaged exponential behavior, meant to reproduce the so-called Samuelson effect. This is an observed feature of prices volatilities, common to many commodity markets, consisting of increasing volatility of prices as time approaches maturity \cite{lautier}. For an analysis of the impact of the Samuelson effect on option pricing, see \cite{schmeck2016pricing} and \cite{schmeck76}. The second factor is independent of time, but varies for contracts with different delivery in a \emph{seasonal} and \emph{no-arbitrage} way (see \cite{BPV,latini}) and accounts for a finer reproduction of the term structure of futures volatilities. We remark that, since our model is additive, by ``volatility'' we mean the parameter (or function of parameters) that determines the variability of prices and not of log-prices as in geometric models.

Building upon \eqref{multifactor}, we assume that the 
stochastic evolution of  a generic 
future price $F(\cdot,T_1,T_2)$, delivering in the period $[T_1,T_2]$, from $t$ to $T$ is described by 
\begin{align}\label{two-factor}
F(T,T_1,T_2)=F(t,T_1,T_2)&+\int_t^T \Gamma_1(u,T_1,T_2)\,dJ_1(u)+\Gamma_2(T_1,T_2) (J_2(T)-J_2(t)),
\end{align}
where 
\begin{align}
\label{gamma1_3}
\Gamma_1(u,T_1,T_2) &:= \frac{1}{T_2-T_1} \int_{T_1}^{T_2} \gamma_1 e^{-\mu (\tau-u)} \, d\tau 
=\frac{\gamma_1(e^{-\mu (T_1-u)}-e^{-\mu (T_2-u)})}{\mu (T_2-T_1)},
\\ \label{gamma2_3}
\Gamma_2(T_1,T_2) &:= \frac{1}{T_2-T_1} \int_{T_1}^{T_2} \gamma_2(\tau) \, d\tau.
\end{align} This model is in the spirit of \cite{BPV,latini}. Thus, the special form of the coefficients arises from the implicitly underlying assumption that  $F$ can be written as the average over an underlying artificial futures price with instantaneous delivery.
By focusing our attention to the first component, the parameter $\gamma_1\in\bR^+$ models the \emph{base volatility}, that is the volatility of contracts resulting from the first component with distant delivery i.e. $T_1-t\to\infty$. This coefficient has an exponential rate given by $\mu\in\bR^+$, which determines an increase of volatility as time approaches delivery i.e. $T_1-t\to0$. This effect is averaged over the delivery period for no-arbitrage arguments (as explained in \cite{latini}) and mimics the Samuelson effect. Regarding the second component, the function $\gamma_2:[0,\infty)\to\bR^+$ models the general seasonal behavior of volatility (it takes high values in periods of high volatility and low values for periods of low volatility).  See \cite{borovkova2017electricity} for a time change model of seasonal volatility in electricity markets. The seasonal function $\gamma_2$ can be specified either in a parametric or nonparametric fashion. For instance, one can have
\begin{equation*}
\gamma_2(\tau) := \gamma_2 + b\tau+ \sum_{j=1}^m \bigl(a_{2j} \cos(\omega j \tau) + a_{2j+1}\sin(\omega j \tau) \bigr), 
\end{equation*}
with $\omega =  2\pi /365$, $m \in \bN$, $b\in\bR$ (for capturing possible deterministic linear trends), $(\gamma_2,a_2,\ldots,a_{2m+1}) \in \bR^{2m+1}$. \\
Note also that the driving L\'evy processes are the same for all contracts. This is also a result of the assumption of an underlying instantaneous futures dynamics (see \cite{BPV}, \cite{latini}). Note that we have fully specified the form of $ \Gamma_1(\cdot,T_\cdot, T_\cdot)$, where the coefficients $\gamma_1$ and $\mu$ are independent of 
the delivery period. With this specification,  the  NOA condition \eqref{eq:NOA2} for $\Gamma_1$ is naturally fulfilled. On the other hand, if the NOA condition has to be satisfied for the chosen form of $\Gamma_1(u,T_i, T_{i+1})$ as in \eqref{gamma1_3}, it is indeed not possible to choose $\gamma_1$ and $\mu$ different for each delivery period. Though, we leave unspecified the coefficient of the second factor, $\Gamma_2$, so to have more modeling flexibility that is able to account for a finer reproduction of the term structure. In the empirical analysis to come, we will use a non-parametric approach and estimate one value for each delivery period. Thus, the following condition has to be fulfilled
\begin{align}
\Gamma_2(T_1,T_{n})&= \frac{1}{T_n-T_1}\sum_{i=1}^{n-1}(T_{i+1}-T_i) \Gamma_2(T_i,T_{i+1})\;.
\end{align}

The coefficients in \eqref{gamma1_3} and \eqref{gamma2_3} modulate the variability of the two stochastic processes $J_1$ and $J_2$, which are both defined as centered versions of NIG L\'evy processes. Let us recall that a L\'{e}vy process is called Normal Inverse Gaussian with parameters $(\alpha,\beta,\delta,\mu)$ if its characteristic triplet is $(\chi,0,\nu)$ with
\begin{align}
\chi &= m + \frac{2\alpha\delta}{\pi} \int_0^1 \sinh(\beta x) K_1(\alpha x)\,dx,
\\\label{levy_measure}
\nu(dy)&= \frac{\alpha\delta}{\pi|y|} K_1(\alpha|y|) e^{\beta y}\,dy,
\end{align} 
where $K_1$ is the modified Bessel function of the third kind with index 1 (in the terminology of \cite[Section 9.6]{AS}), $0\leq|\beta|<\alpha$ and $\delta>0$, $m\in\bR$ (see \cite{BN}). 
Given a NIG process $L$, the random variable $L(1)$ is NIG distributed with parameters $(\alpha,\beta,\delta,\mu)$. The NIG distribution is a subclass of a very flexible family, the Generalized Hyperbolic distributions, and it can accomodate heavy-tails and skewness. The parameter $\alpha$ rules the tail heaviness of the distribution, $\beta$ determines the skewness, $\delta$ is a scale parameter and $\mu$ indicates the location of the distribution. 

In general, a NIG process is not centered. Therefore, in order to define $J_1$ and $J_2$ we subtract from two general NIG processes $L_1,L_2$ the corresponding expected value (multiplied by time). For $j=1,2$, let $L_j$ be a NIG L\'{e}vy process with parameters $(\alpha_j,\beta_j,\delta_j,m_j)$ and characteristic triplet $(\chi_j,0,\nu_j)$ and set 
$$
J_j(t) = L_j(t) - \bE[L_j(t)] = L_j(t) - t \left(\chi_j+\int_{|y|\geq 1} y\,\nu_j(dy)\right). 
$$
Then, $J_j$ is a centered NIG process. In particular, it can be easily shown (cf. the characteristic function of $J_j$ in \eqref{ch_fun_center} to see this) that $J_j$ is a NIG process with parameters $\left(\alpha_j,\beta_j,\delta_j,-\frac{\delta_j\beta_j}{\sqrt{\alpha_j^2-\beta_j^2}}\right)$.

As already mentioned, in order to compute the option prices, we will need the characteristic function of $F(T,T_1,T_2)$. This in turn reduces to finding the characteristic function of $Z(t,T,T_1,T_2) := F(T,T_1,T_2)-F(t,T_1,T_2)$. We compute it in the appendix in order not to make the presentation unnecessarily heavy.

\section{Option pricing for additive models by Fourier transform methods}

\label{sec:3}

We consider the pricing of European vanilla options written on futures contracts of the type introduced in the previous section. We discuss the method for call options, being the case of puts completely analogous. 
Let $C(t;T,K,T_1,T_2)$ be the price at time $t$ (the observation date) of a call option, with strike price $K$ and exercise time $T$, that is written on a futures contract with delivery period $[T_1,T_2]$. By no-arbitrage (see, for instance, \cite{MR2057475}), 
\begin{equation}\label{call}
C(t;T,K)=\bE\left[ (F(T)-K)_+\left| \mathcal{F}_t \right. \right],
\end{equation}
where $F(T)$ is the price of the underlying futures at the option exercise as in \eqref{multifactor} and $\mathcal{F}_t$ represents the filtration at time $t$, i.e. the information flow up to time $t$. In order to ease the notation, sometimes we will not write the dependence on the delivery period, which does not come into play in this discussion. We recall that the expectation is taken under the risk-neutral measure $\bQ$, even though it is not explicitly indicated, and throughout this work $\bQ$ will be the only probability measure we will deal with. 
By applying the definition of conditional expectation, we can write
\begin{equation}\label{density}
C(t;T,K)=\int_{K}^{+\infty} (s-K)\,q_{t,T}(s)\,ds,
\end{equation}
where $q_{t,T}$ is the (risk-neutral) density function of $F(T)$ conditioned up to time $t$. This formula yields an expression for the option value, for instance, when the distribution of the underlying $F$ allows for an explicit formula for the density, that, moreover, can be integrated against the payoff function in a tractable way (as it is the case in the Black model \cite{black}, where the underlying follows a geometric Brownian motion). 

We follow the alternative approach of \cite{CM}, which consists, roughly speaking, in computing the Fourier transform of \eqref{call} as a function of $K$ (after proper manipulations) so to recover the option value by inverse transform. This has been studied by several authors and it has been shown to be a very convenient way to compute option prices in the case that the characteristic function of the underlying is known explicitly, while the density is not. 
The starting point of the above mentioned approach is the observation that, as $K$ goes to $-\infty$, $C(t;T,K)\to\infty$, so that in particular $C(t;T,K)$ is not integrable as a function of $K$ for large negative values. This means that the option value $C(t;T,K)$ does not satisfy the assumptions required for computing its Fourier transform. In order to overcome this, we follow \cite{CM}, who suggest two approaches that we here recall and apply to the case of additive models. 

The first approach requires to modify the option value with a damping term $e^{a K}$, where $a > 0$ in our case should be such that $\bE[e^{a Z(t,T,T_1,T_2)}] < + \infty$. While this does not pose problems in Gaussian models, where any $a > 0$ would satisfy this condition, in our case this condition will depend on the NIG parameters, still unknown. Moreover, even if theoretically any ``good" $a$ would deliver the same result, it is a well-known result that ``extreme" values for $a$ (i.e., too close to 0 or to the upper bound) give numerical instabilities: thus, one should also optimize with respect to $a$ in order to have good numerical results.

 For the arguments above, we instead choose to use
the second approach,  which 
consists in substracting the time value of the option\footnote{ For the sake of completeness, we performed the calibration of Section 5 also with the first approach, which led to results analogous to those that we present here, but with the additional complications mentioned above.}. 
Following \cite{CM} (see also \cite{MR2042661}), define the modified time value of the option (as a function of the strike $K\in\bR$) by 
\begin{equation}\label{mod_time}
z^{MT}_{t,T}(K)=C(t;T,K)-(F(t)-K)_+,
\end{equation}
and, if it is square-integrable, compute its Fourier transform
\begin{equation}\label{xi_mt}
\xi^{MT}_{t,T}(v)=\int_{-\infty}^{+\infty} e^{ivK}\,z^{MT}_{t,T}(K)\,dK.
\end{equation} 
The modified option price \eqref{mod_time} is then recovered by Fourier inversion after integrating by parts: this is derived in detail in the appendix.

\section{Calibration procedure}\label{sec:4}

As we have derived in the previous section semi-analytical expressions (analytical up to integration) for the option prices, we now move to discussing a calibration methodology that we are going to apply in our empirical study. 
First, we discretize the option value that is given in integral form in 
\eqref{z_inv_re2} (see the appendix) 
in the domain of integration. Then, we select a finite grid of strike prices, that consists in practice of the listed options available in the market for a given underlying. This procedure reduces the valuation problem to the computation of a finite sum of vectors, where each component is the option price for a given strike. Then we introduce a least squares problem designed to find the parameters that minimize an error function (for a discussion about the choice of the error function, we refer to \cite[Section 5.4]{borger}). We can compute it for two possible quantities, the model error on option prices and on the IVs given by Black's formula. By minimizing on IVs rather than option prices, we weight at the same way contracts with different maturity. However, minimizing on prices is preferable in terms of speed of computation, as it does not require the inversion of Black's formula at each step.  
We separate the calibration routine in three cases, of increasing generality, in order to discuss the presence of different constraints case-by-case. 

\subsection{Discretization of model option prices}

The quantity that we have to discretize in \eqref{z_inv_re2} takes the following form: 
\begin{equation}\label{truncation}
z_{t,T}(K)=\frac{1}{\pi} \int_{0}^{+\infty} \mathrm{Re}\left( e^{-iKv}\,\xi_{t,T}(v)\right)\,dv.
\end{equation}
First, we choose an upper limit $A\in\bR^+$ in the above integration (see \cite[Section 3.1]{CM} for a discussion on how to do this optimally). Then, if we apply a simple  Euler rule to the truncated integral, we find an expression of the form
\begin{equation}\label{euler}
z_{t,T}(K)\approx\frac{1}{\pi} \sum_{j=1}^{N} \mathrm{Re}\left(e^{-i v_j K}\,\xi_{t,T}(v_j)\right)\eta,
\end{equation}
where $N\in\bN$, $\eta := A/N$ and $v_j:=\eta(j-1)$ is the integration step. Given $M\in\bN$ different strikes with granularity $\kappa>0$, we compute the function in \eqref{euler} for the following values of $K$:
$$
K_u:=\overline K+\kappa(u-1), \qquad\mbox{for } u=1,\ldots,M,
$$
being $\overline K$ the lowest strike price traded. By plugging this in \eqref{euler}, we get for $u=1,\ldots,M$ 
\begin{equation}\label{z_discrete}
z_{t,T}(K_u)\approx w_{t,T}(K_u):=\frac{1}{\pi} \sum_{j=1}^{N} \mathrm{Re}\left(e^{-i \kappa\eta(j-1) (u-1)} e^{-i \overline K \eta(j-1)}\,\xi_{t,T}(\eta(j-1))\right)\eta.
\end{equation}

As pointed out in \cite{CM}, this formula is suitable for the application of the fast Fourier transform (FFT) algorithm. In order to do this, one must impose that the number of strike prices considered is equal to the number of integration nodes, i.e. $M=N$ (which is typically chosen as a power of 2). Moreover, it must hold that   
$$
\kappa\eta=\frac{2\pi}N,
$$
which consists of a trade-off between the grid for the integration and the granularity of strike prices. In particular, since in practice the granularity $\kappa$ is given, this equality  univocally defines the integration grid as a function of $N$. However, in our application we will not make use of the FFT algorithm and so, in particular, we will not impose the above restrictions on $N,M,\kappa,\eta$. This is motivated by the fact that we do not have a significant advantage in the computational complexity of the problem, being the number of strikes consistently lower than $N$. Furthermore, the focus of our work is not on the speed-up of the calibration procedure, but rather on empirical results, so that we do not exclude the possibility to apply the FFT algorithm, being still possible from a theoretical point of view. 

\subsection{Parameters estimation}

Since we have at disposal formulas in discrete form that can be readily implemented, we can now discuss how to fit the two-factor model to market data. 
The calibration happens statically, taking a snapshot of the market, in the sense that we fix a trading date and observe the market option prices for that date. 
We then find the parameters that fit a certain distance from theoretical to model prices best.  
As mentioned at the beginning of this section, one can alternatively use IVs instead of option prices. The futures model under consideration is defined in a such a way that there is no possibility of arbitrages, also in the case of overlapping delivery periods. No-arbitrage implies certain relations on the coefficients, that translate into parameter constraints. In order to highlight this, we present the objective function and the parameters set for the following three cases: single underlying, many underlyings but non-overlapping delivery periods, and the general case of possibly overlapping delivery periods. 

\subsubsection{Black's formula}

The first and most used in practice model for option prices written on futures is the Black model \cite{black}. It assumes that the underlying follows a geometric Brownian motion as in the Black-Scholes formula. Also, the expression for the call price is very similar to the Black-Scholes one, with the futures price replacing the stock price, but with different discounting. Since we are assuming that the risk-free rate is zero, the Black-Scholes and the option price given by the Black model are actually the same in our case. We recall the Black formula here because, in addition to use it as a benchmark in the upcoming empirical application, we use it to compute the implied volatility of both market and model prices:
\begin{equation}
C_{BS}(t;T,K) = F(t,T)N(d_1)-KN(d_2),
\end{equation}
where $N$ denotes the cumulative distribution function of a standard Normal random variable and
$$
d_1 = \frac{\log{\frac{F(t)}{K}}+\frac12 \sigma^2(T-t)}{\sigma\sqrt{T-t}},\qquad d_2= d_1 - \sigma\sqrt{T-t}.
$$
The implied volatility of a given option with price $P$ and strike $K$ is defined as the only (positive) number $\sigma$ such that the Black formula for a strike $K$ and volatility $\sigma$ (all other quantities being equal) gives the price $P$. The two-factor model aims to reproduce the IV profile of market option prices, that is the plot of $\sigma$ with respect to $K$. Black's formula yields constant implied volatility with respect to $K$, while real option prices usually display smiles or smirks (i.e. the IV is not constant and shows a certain convexity). 

\subsubsection{Single  underlying contract}\label{sec:sc}

Assume that we observe at a certain date $t<T$ the prices $c_*(t,T,T_1,T_2,K_u)$ of a call option for $M$ different strike prices $K_u$, with exercise at time $T$  written on a futures with delivery over $[T_1,T_2]$. 
To fit the model to the observed prices we search for the parameters that minimize a certain error function. Specifically, we introduce the following least-squares problem: 

\begin{equation}
\widehat \theta :=\underset{\theta\in\Theta}{\arg\min} \sum_{u=1}^M |c(t,T,T_1,T_2,K_u)-c_*(t,T,T_1,T_2,K_u)|^2,
\end{equation}
where 
the set $\Theta$ contains all the parameters appearing in the approximation formula of the model option price $C(t,T,T_1,T_2,K_u)\approx c(t,T,T_1,T_2,K_u)$ where 
\begin{equation}\label{call_semi}
c(t,T,T_1,T_2,K_u):=
w_{t,T}(K_u)+ (F(t,T_1,T_2)-K_u)^+
\end{equation}
where $w_{t,T}(K_u)$ for $u=1,\ldots,M$ is the discrete function in \eqref{z_discrete}.

If we assume that $F(T,T_1,T_2)$ is given by a two-factor pure-jump model of the form \eqref{two-factor}
where $\Gamma_i$ are defined in \eqref{gamma1_3} and \eqref{gamma2_3} and $J_i$ is a centered NIG L\'{e}vy process with parameters $(\alpha_i, \beta_i, 1)$, then $\Theta=\{\theta=(\alpha_1,\alpha_2, \beta_1,\beta_2,\mu,\gamma_1,\Gamma_2(T_1,T_2))\in(\bR^+)^2\times(\bR^+_0)^4\times\bR^+ : 0\leq|\beta_j|<\alpha_j\}$. We do not need to indicate the parameter $m$ of the original NIG L\'{e}vy process $L_i$ (see Section 3.2) since it does not appear in the centered version. Also, the parameter $\delta$ is assumed to be $1$ without here because of the presence of the multiplying factors $\Gamma_1(u,T_1,T_2)$ and $\Gamma_2(T_1,T_2)$, thanks to the property that, given a NIG$(\alpha, \beta, \delta)$ distributed random variable $X$, for a constant $\gamma>0$, $\gamma X$ is distributed as a NIG$(\alpha/\gamma, \beta/\gamma, \delta\gamma)$. This means that letting $\delta$ vary in the parameters set would result in an overparametrization of the minimization problem.

\subsubsection{Non-overlapping futures}

Let us assume that we are at time $t$ and observe $I$ call options written on futures contracts with non-overlapping delivery periods $[T_{i,1},T_{i,2}]$ (for example Jan/YY, Feb/YY, Mar/YY), exercise at $T_i$ and $M_i$ strike prices 
$$
K^i_{u} :=\overline K_i+\kappa_i(u-1),
$$ 
for $u=1,\ldots,M_i$ and $i = 1,\ldots,I$. 
In the case of more than one contract, we introduce the following least-squares problem:
\begin{equation}\label{estimator}
\widehat \theta :=\underset{\theta\in\Theta}{\arg\min} \sum_{i=1}^I \sum_{u=1}^{M_i} |c(t,T_i,T_{i,1},T_{i,2},K^i_u)-c^*(t,T_i,T_{i,1},T_{i,2},K^i_u)|^2.
\end{equation}
where $\Theta$ is now the set of parameters $\{\theta=(\alpha_1,\alpha_2, \beta_1,\beta_2,\mu,\gamma_1,\{\Gamma_i\}_{i=1}^I)\in(\bR^+)^2\times(\bR^+_0)^4\times(\bR^+)^I : 0\leq|\beta_j|<\alpha_j\}$.
For $I=1$ we recover exactly the previous case. 

\subsubsection{Whole market}\label{sec:wholemarket}

In general options traded in power markets are written on $I$ futures contracts with different delivery length and possibly overlapping. For example, one can have Apr/YY, May/YY, Jun/YY, Jul/YY, Q2/YY, Q3/YY, Cal-YY simultaneously traded. As explained in \cite{latini}, in order to estimate the parameters in a consistent, i.e. arbitrage-free, way, we have to take into account additional constraints on the parameters. From Equation \eqref{gamma2_3}, it can be shown that the parameters $\Gamma_i$ satisfy the following constraints
\begin{equation}\label{96}
\Gamma_i = \Gamma_2(T_{i,1},T_{i,2}) = \sum_{j=1}^{n} \frac{T_{j,2} - T_{j,1}}{T_{i,2} - T_{i,1}} \Gamma_2(T_{j,1},T_{j,2}) = \sum_{j=1}^{n} \frac{T_{j,2} - T_{j,1}}{T_{i,2} - T_{i,1}} \Gamma_j \end{equation}
whenever $[T_{i,1},T_{i,2}]$ is the union of disjoint intervals $[T_{j,1},T_{j,2}]$ for $j=1,\ldots,n$, i.e. for all the forwards with overlapping delivery, see Section \ref{sec:2} for details. Let us call \emph{atomic}, the contracts whose delivery period can not be partitioned by the delivery periods of other futures. In other words, we suppose that $m$ forwards $F_1,\ldots,F_m$ have non-overlapping delivery periods $[T_{1,1},T_{1,2}],\ldots,[T_{m,1},T_{m,2}]$ and such that the delivery periods of the other contracts traded in the market can be expressed as union of the former.  For example, assume that we observe in the chosen calibration window the option prices for futures contracts delivering over Apr/YY, May/YY, Jun/YY, Jul/YY, Q2/YY, Q3/YY, Cal-YY. On one hand, Q2/YY is not atomic, since it can be ``splitted'' into Apr/YY, May/YY and Jun/YY. On the other hand, Q3/YY turns out to be atomic, even if Jul/YY is already traded, as Aug/YY and Sep/YY are not observed. For the same reason, Cal-YY is considered atomic as well. Then, in this example, if $\Gamma_1,\ldots,\Gamma_{7}$ denote the corresponding parameters, we have that $\Gamma_1,\Gamma_2,\Gamma_3,\Gamma_4,\Gamma_5,\Gamma_7$ are free parameters, as they refer to atomic contracts, whereas to determine $\Gamma_5$ we use Equation \eqref{96}. Consequently, we define the same statistics as in \eqref{estimator} but where now the vector of parameters is subject to the additional constraints given by Equation \eqref{96}. For example, with the convention that $\Gamma_{\mathrm{Q2/YY}}$ denotes the parameter $\Gamma_i$ corresponding to the contract Q2/YY, then
\begin{equation}\label{94}
\Gamma_{\mathrm{Q2/YY}} = u_{\mathrm{Apr/YY}} \Gamma_{\mathrm{Apr/YY}} +u_{\mathrm{May/YY}} \Gamma_{\mathrm{May/YY}} + u_{\mathrm{Jun/YY}} \Gamma_{\mathrm{Jun/YY}},
\end{equation}
where the weights $u_i$ are defined according to the number of days in the month/quarter (e.g. for Apr/YY we have $u_{\mathrm{Apr/YY}}=30/91$).

Thus, we summarize the optimization problem with NOA-conditions:

\begin{equation}\label{estimator2}
\widehat \theta :=\underset{\theta\in\Theta}{\arg\min} \sum_{i=1}^I \sum_{u=1}^{M_i} |c(t,T_i,T_{i,1},T_{i,2},K^i_u)-c^*(t,T_i,T_{i,1},T_{i,2},K^i_u)|^2.
\end{equation}
where $\Theta$ is now the set of parameters $\theta=(\alpha_1,\alpha_2, \beta_1,\beta_2,\mu,\gamma_1,\{\Gamma_i\}_{i=1}^I)\in(\bR^+)^2\times(\bR^+_0)^4\times(\bR^+)^I $ such that 
\begin{align*}
 0\leq|\beta_j|<\alpha_j,\;\\
 \Gamma_i =  \sum_{j=1}^{n} \frac{T_{j,2} - T_{j,1}}{T_{i,2} - T_{i,1}} \Gamma_j\;.
\end{align*} 

\newpage 

\section{Empirical study}\label{sec:5}

In this section we describe our dataset, compute and discuss the empirically observed volatility implied by settlement prices. We compare the performance of three models: a purely Gaussian model,  a one-factor model (being a special case of the two-factor one) and the general two-factor model.  

\subsection{Data description}

The contracts that we consider in our application are European-styled call options traded at the EEX Power Derivatives market. The underlying assets are futures contracts that prescribe the delivery of 1 MWh for each hour of each day of a month, a quarter or a year. More specifically, we will consider call options written on the Phelix Base index of the German/Austrian area. These options are called Phelix Base Month/Quarter/Year Options and the EEX official product codes are O1BM, O1BQ, O1BY. The term \emph{Base} refers to \emph{Base Load}, because the delivery of electricity takes place for each hour of the day, in contrast to \emph{Peak Load} contracts, that instead prescribe the delivery only for the hours from 8 to 20. Usually, the exercise of the options under consideration is few trading days before the start of delivery. Since recently, yearly options are available for four different exercise dates. We consider in our dataset only the ones expiring few days before delivery, in analogy to quarterly and monthly contracts. Since there is not enough liquidity in the market in order to extract information on the IV surface from traded market quotes, we consider the settlement prices, that are available for a sufficiently large range of maturities and strikes: though settlement prices do not represent trades that really take place in the market, they contain information on market expectations.  
We observe the market for a representative day: Monday, March 5, 2018. For each option, we consider the strike prices in the range 90\%--110\% of the underlying current price (as, for example, in \cite{arismendi,back}). At this date the listed options with available settlement prices are the ones written on five monthly (Apr/18, May/18, Jun/18, Jul/18, Aug/18), six quarterly (Q2/18, Q3/18, Q4/18, Q1/19, Q2/19, Q3/19) and three yearly (Cal-19, Cal-20, Cal-21) futures.\\

Before moving on with the calibration procedure, let us comment on the empirical market-implied volatilities first, which we have calculated by inverting Black's formula. They are displayed in Figure \ref{fig:surface}. We observe a very pronounced smile for the contracts with forthcoming delivery, i.e. the contracts with delivery in April 18 and in the second quarter. Furthermore we observe that the farther the  beginning of the delivery period is in the future, the less pronounced is the smile: for the next beginning periods in May, June and July the smile is present but less and less pronounced, and already in August and the Q3 contract there is still a tendency to smile, but not pronounced at all. 
Generally we observe a \emph{forward skew} (i.e. higher IVs for out-of-the-money calls). This can be interpreted as a ``risk premium'' paid by option buyers for securing supply. Remember however, that we consider only strikes  in a range of about $90 \% - 110 \%$ of at-the-moneyness due to illiquidity considerations.

Remember that, in our model, the smile is produced by the presence of jumps. More precisely, we consider two independent  NIG processes  as the stochastic drivers of the two factors: the first factor captures the short time behavior modeled by $\Gamma_1(u,T_1,T_2)$ defined in Equation \eqref{gamma1_3}, also referred to as the Samuelson effect; the second factor, modeled by   $\Gamma_2(u,T_1,T_2)$ as in \eqref{gamma2_3}, captures seasonal behavior depending on both the delivery period and the no-overlapping-arbitrage condition \eqref{94}. 

From the observation of the market implied volatilities, we have some expectations on what the estimated parameters of our futures model should 
 be. First of all, we should get better estimates incorporating the Samuelson-factor, as we do observe a different smile behavior close to delivery and far away from delivery. We observe a very pronounced smile shortly before delivery, and a less pronounced smile far away from delivery. Consequently, the estimated NIG parameters of the Samuelson factor should be such that the resulting distribution is far away from a Normal distribution.  Furthermore, the estimated NIG parameters of the seasonal factor should be such that they are closer to a normal distribution.

\subsection{Calibration}

We perform the calibration procedure described in Section 3.4.3, first for the one-factor model derived by \eqref{two-factor} by setting the first coefficient to 0, i.e. $\Gamma_1(u,T_1,T_2)\equiv0$, and then for our general two-factor model. We compare both models to the empirical IVs and the one (constant across strikes) generated by Black's model, that we estimate with the same procedure of the other models (except for the computation of the price, which can be computed analytically for the Black model). We show the results for the minimization on market prices (see Section 4.4), since the calibration is faster and, even though the case of IVs yields by definition a lower residual, it gives very similar results. 

After numerical experiments with test parameters, the integral in \eqref{truncation} has been truncated at $A=10$ and computed by an adaptive Simpson quadrature rule already implemented in MATLAB (as in \cite{nomikos}), which takes into account the oscillatory behavior of the integrand (cf. \cite{CM}). 

We find that, in general, the optimization routine falls into local minima. However, by selecting a starting condition that is ``sufficiently close'' to the observed IVs, the minimization converges (cf. \cite{CT} for well-posedness of this kind of problems). 
As a by-product from the estimation of the one-factor model, we derive that, under the risk-neutral measure, futures prices are leptokurtic and have significantly positive skewness. This is reflected also into the shape of empirical IVs, which,  as we already mentioned, display a \emph{forward skew} (i.e. higher IVs for out-of-the-money calls). 
 The IVs of market and model prices are plotted in Figures \ref{fig:1_ivm1}--\ref{fig:3_iv}.

\subsection{Results}

\begin{figure} 
	\begin{subfigure}{0.5\textwidth}
		\includegraphics[width=\linewidth]{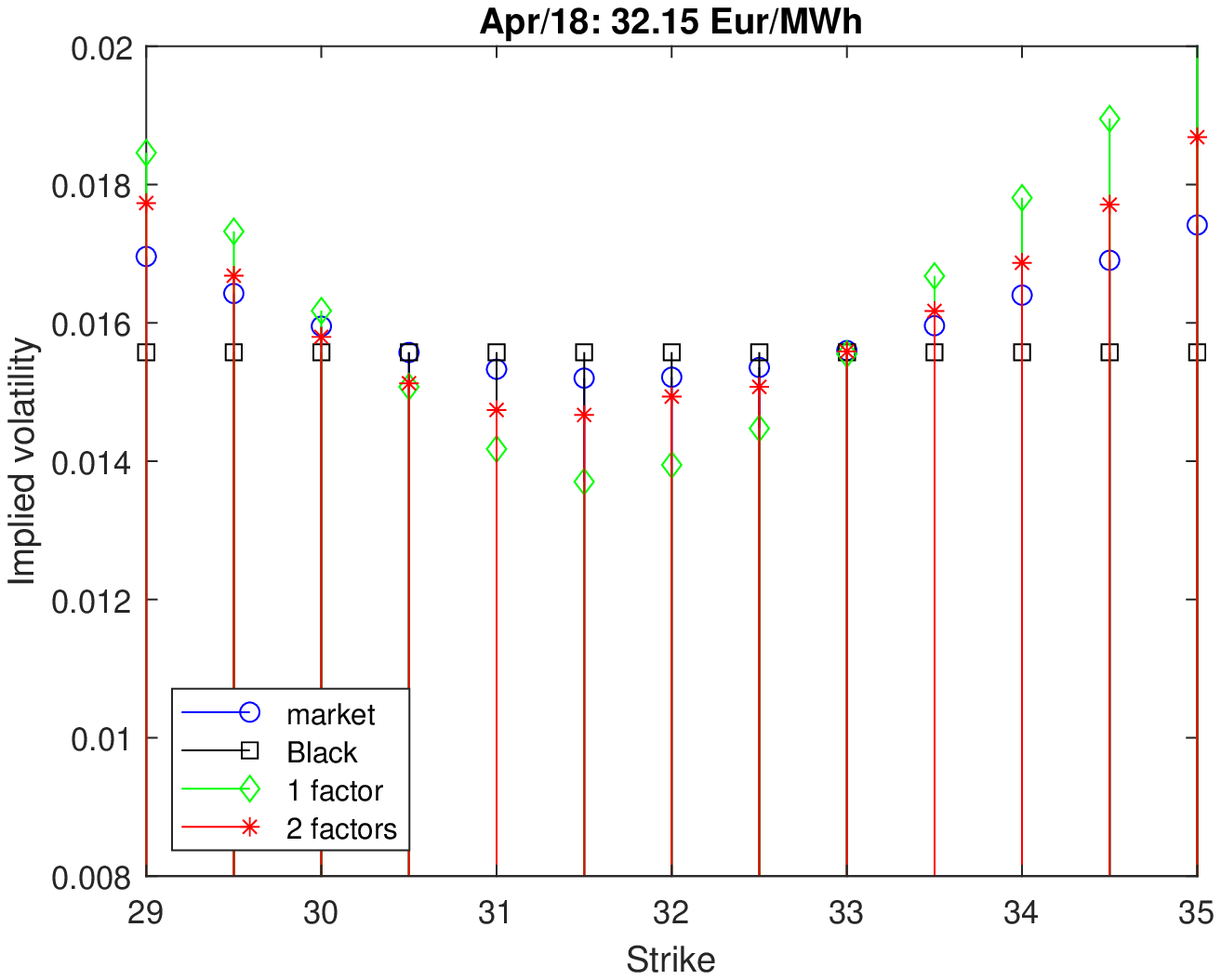}
	\end{subfigure}\hspace*{\fill}
	\begin{subfigure}{0.5\textwidth}
		\includegraphics[width=\linewidth]{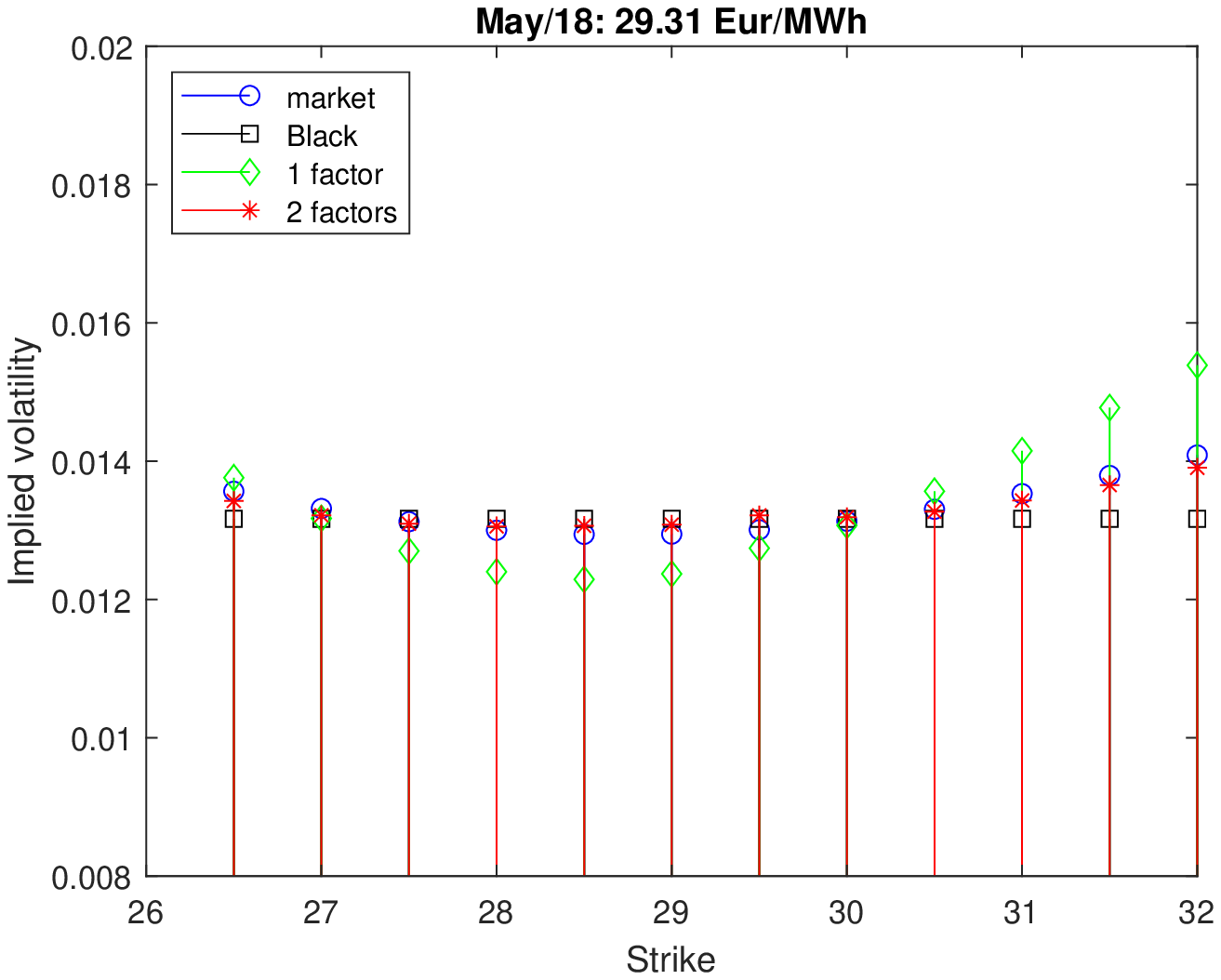}
	\end{subfigure}
	\caption{Selection of monthly delivery periods. Implied volatility for the Black, one-factor and two-factor model compared to the empirical implied volatilities of  options listed at March 5, 2018 (the corresponding underlying current price is indicated above each plot).} \label{fig:sel}
\end{figure}

We first address the estimated parameters for the futures model, before we comment on the resulting fit of market implied volatilities.
We have estimated the distribution of the driving NIG process. Remember though that we consider centered L\'evy processes, thus $\mu=0$ and $\delta=1$ due to the scaling property of the NIG distribution (see Section \ref{sec:sc}). The estimated parameters are reported in Table \ref{tab:parameters_1f} and \ref{tab:parameters} for the one factor and two factor model respectively. 

\begin{table}[htb]
	\centering
	\begin{small}
		\begin{tabular}{rrrrrrrr}
\hline
			$\alpha$ & $\beta$ &  $\Gamma_2(1)$ &  $\Gamma_2(2)$ & $\Gamma_2(3)$ & $\Gamma_2(4)$ & $\Gamma_2(5)$ & $\Gamma_2(6)$ \\
			0.0059 & 0.0019 & 0.0464 & 0.0327 & 0.0315 &  0.0311 &  0.0293 & 0.0368 \\
\hline
  $\Gamma_2(7)$ &  $\Gamma_2(8)$ &  $\Gamma_2(9)$ &  $\Gamma_2(10)$ &  $\Gamma_2(11)$ & $\Gamma_2(12)$ & $\Gamma_2(13)$ & $\Gamma_2(14)$  \\
			0.0284  & 0.0310  & 0.0304 & 0.0211  & 0.0209 & 0.0271 & 0.0255 & 0.0244 \\
\hline
		\end{tabular}
	\end{small}
	\caption{Calibrated parameters for the one factor model at 2018, March 5th.}
	\label{tab:parameters_1f}
\end{table}

\begin{table}[htb]
	\centering
	\begin{small}
		\begin{tabular}{rrrrrrrrrrrrrrr}
\hline
			$\alpha_1$ & $\beta_1$ & $\alpha_2$ & $\beta_2$ & $\gamma_1$ & $\mu$ &  $\Gamma_2(1)$ &  $\Gamma_2(2)$ & $\Gamma_2(3)$ & $\Gamma_2(4)$ \\
			0.1890 & 0.0586 & 0.0005 & 0.0002 & 0.1656 &  0.0044 &  0.0129 & 0.0054 & 0.0060 & 0.0068\\
\hline
			$\Gamma_2(5)$ & $\Gamma_2(6)$ &  $\Gamma_2(7)$ &  $\Gamma_2(8)$ &  $\Gamma_2(9)$ &  $\Gamma_2(10)$ &  $\Gamma_2(11)$ & $\Gamma_2(12)$ & $\Gamma_2(13)$ & $\Gamma_2(14)$  \\
			0.0064  & 0.0081  & 0.0066 & 0.0091  & 0.0093 & 0.0055 & 0.0057 & 0.0093 & 0.0084 & 0.0078\\
\hline
		\end{tabular}
	\end{small}
	\caption{Calibrated parameters for the two factor model at 2018, March 5th.}
	\label{tab:parameters}
\end{table}

In Figure \ref{fig:density} we have plotted the theoretical NIG density together with the corresponding normal density. We see that the Samuelson factor is far away from being Gaussian, exhibiting rather fat tails. Also the estimated density of the seasonal component clearly deviates from the Gaussian shape, but  in a less pronounced way. This confirms our expectations and justifies our choice of pure jump L\'evy processes. The seasonal component has a tendency to model  more ``normal" movements, while the Samuelson component is able to account for rare, bigger movements.  This analysis is complemented by Table \ref{tab:moments}, where we have presented  the estimation of the moments of the NIG driving factors. Note that these are the (risk-neutral) moments of the jump drivers $J_1(1)$ and $J_2(1)$ only, i.e. without considering the coefficients and for fixed time equal to 1. By accounting for the coefficients and integrating in time, one would get the (risk-neutral) moments of the futures prices. 

\begin{figure}[htb] 
	\begin{subfigure}{0.5\textwidth}
		\includegraphics[width=\linewidth]{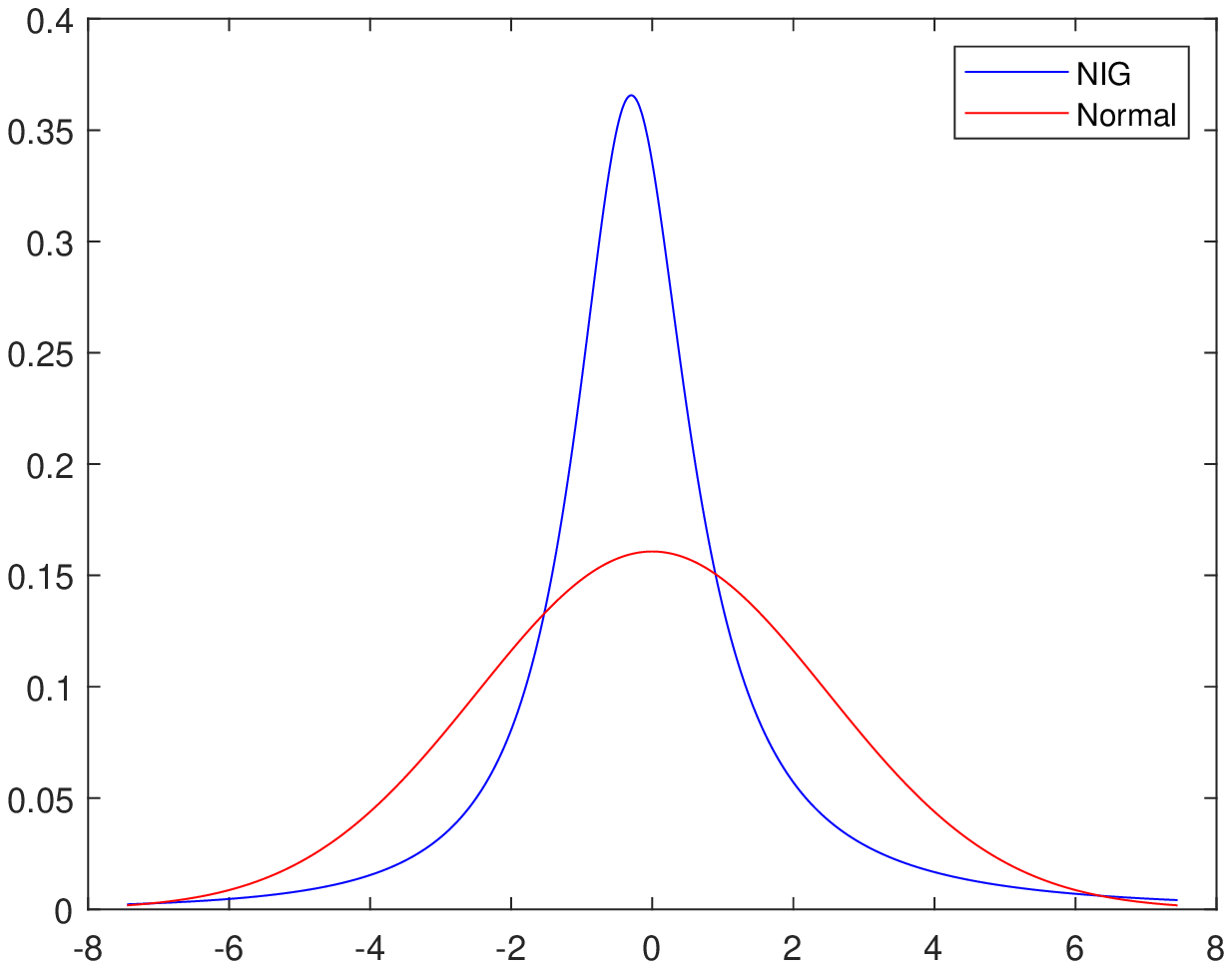}
	\end{subfigure}\hspace*{\fill}
	\begin{subfigure}{0.5\textwidth}
		\includegraphics[width=\linewidth]{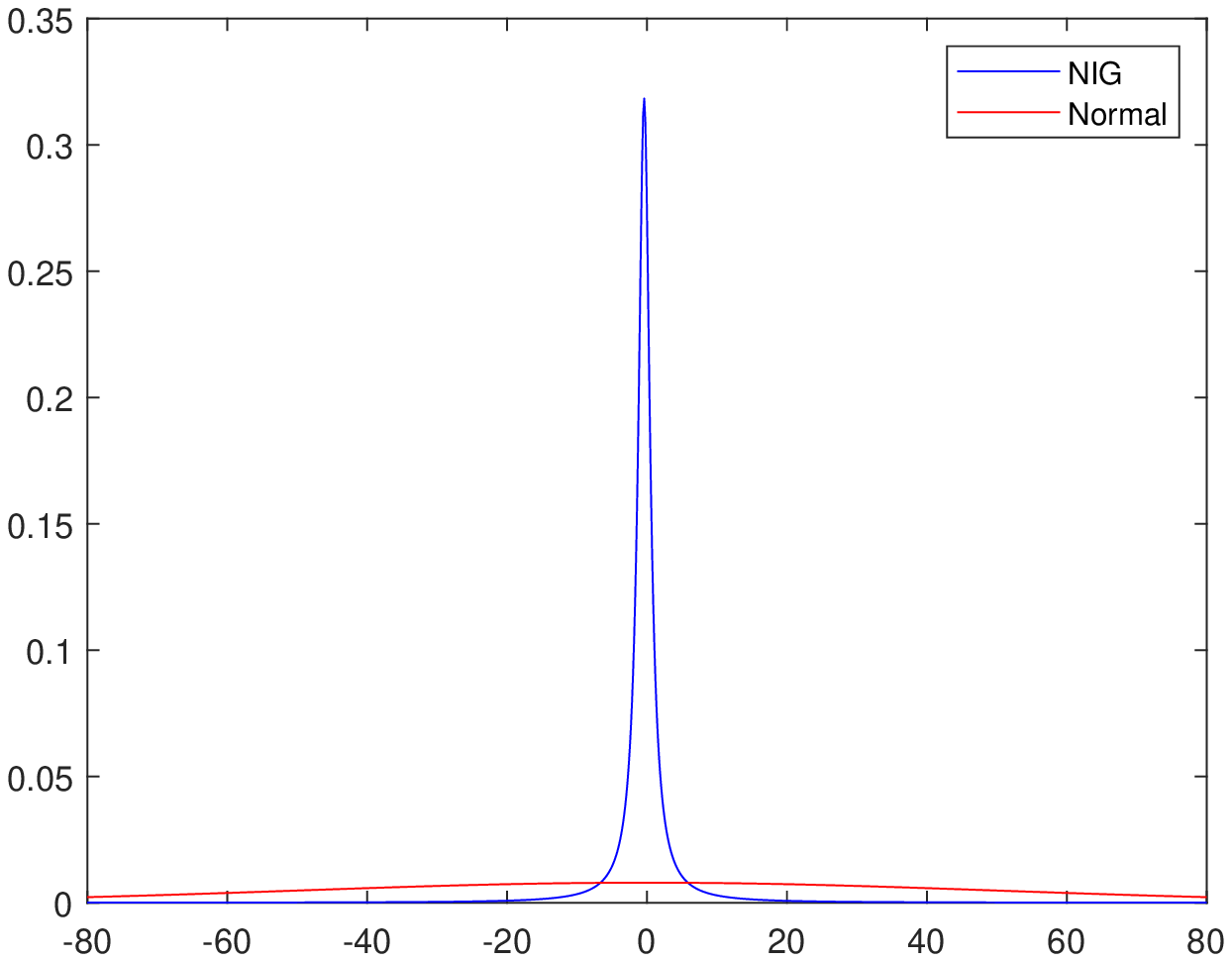}
	\end{subfigure}
	
	\caption{Densities of the NIG factors $J_1$ (left), $J_2$ (right) compared to the corresponding Gaussian densities (i.e. Gaussian density with same mean and variance).} \label{fig:density}
\end{figure}

\begin{table}[htb]
	\centering
	\begin{small}
		\begin{tabular}{crrrrrrrrrrrrrrr}
{\sc component}		& {\sc mean} & {\sc variance} & {\sc skewness} & {\sc excess kurtosis} \\
\hline
$J_1$		& 0   & 6.1603  &  2.1964  & 23.1341 \\
$J_2$		& 0 & 2508.1 & 53.197 & 10192 \\
		\end{tabular}
	\end{small}
	\caption{Moments of the NIG driving factors $J_1$ (top) and $J_2$ (bottom).}
	\label{tab:moments}
\end{table}

We now start to discuss the remaining estimates, that are reported in Table \ref{tab:parameters_1f} and \ref{tab:parameters} for the one factor and two factor model respectively. We would like to draw the attention to the estimates $\Gamma_i$ of the overlapping delivery periods that are restricted through the NOA-condition in \eqref{96}. For our dataset, this is Apr/18, May/18, Jun/18,  as well as Q2/18. The corresponding estimates are  illustrated in Figure \ref{fig:gamma22}.

\begin{figure}
	\centering
	\includegraphics[width=0.7\linewidth]{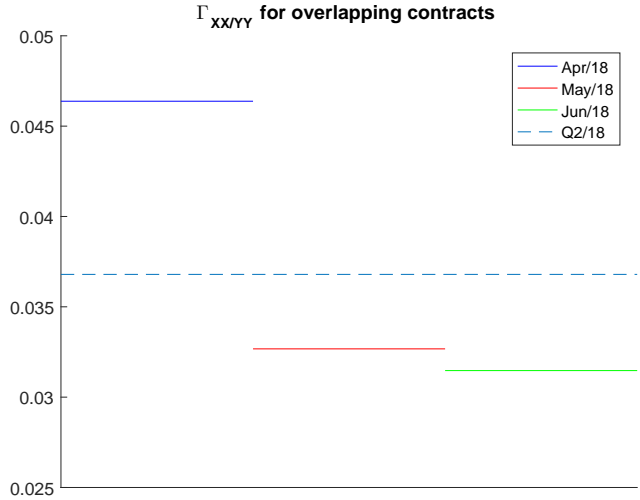}
	\caption{$\Gamma_2$ for overlapping monthly delivery (solid line) and quarterly delivery (dashed lines).}
	\label{fig:gamma22}
\end{figure}


After addressing the properties of our estimated futures model, we move on to discussing the fit of volatility smiles. We divide the discussion depending on the time to maturity -- forthcoming delivery periods (Apr/18, Q2/18), delivery periods that are not immediately forthcoming (May/18, Jun/18, Jul/18, Q3/18, Q4/18)  and delivery periods that are far away (in terms of the number of forthcoming periods). The implied volatility smile of a selection of the delivery periods is plotted in Figure \ref{fig:sel}, both for the 1-factor and 2-factor model.  For comparison, we have added the calibrated volatility  resulting from Black's model, whose values are reported in Table \ref{tab:parameters_BS}. While the selection in Figure \ref{fig:sel} includes the two forthcoming months,  the plots of all smiles that we consider can be found in the Appendix \ref{app:figures} together with the corresponding option prices, see Figures \ref{fig:1_ivm1} - \ref{fig:3_iv}. Here, we have also displayed the smiling volatility surface of monthly, quarterly and yearly delivery periods (see Figure \ref{fig:surface}). 

\begin{table}[htb]
	\centering
	\begin{small}
		\begin{tabular}{rrrrrrrr}
\hline
			$\sigma(1)$ &  $\sigma(2)$ & $\sigma(3)$ & $\sigma(4)$ & $\sigma(5)$ & $\sigma(6)$ & $\sigma(7)$ \\
		0.0156 & 0.0132 & 0.0123 & 0.0121 & 0.0120 &  0.0137 &  0.0109 \\
\hline
   $\sigma(8)$ &  $\sigma(9)$ &  $\sigma(10)$ &  $\sigma(11)$ & $\sigma(12)$ & $\sigma(13)$ & $\sigma(14)$  \\
		 0.0106  & 0.0106 & 0.0094  & 0.0094 & 0.0106 & 0.0104 & 0.0100 \\
\hline
		\end{tabular}
	\end{small}
	\caption{Calibrated parameters for the purely Gaussian model on 2018, March 5th.}
	\label{tab:parameters_BS}
\end{table}

For the forthcoming contracts, both the 1-factor and 2-factor model underestimate the implied volatility at the money, and overestimate it in and out of the money. Nevertheless, the two factor model does well, and captures especially the shape out of the money. It is our winner here.  On the contrary, the one factor model mismatches so much that it is hard to call it satisfactory.
Consider now the contracts that are not immediately forthcoming. Here we find an almost perfect fit of the two factor model, which is very satisfactory; however, the one factor model seems to perform better than before, too.  Finally, for the contracts with delivery period in the far future, both factor models give good results. Nevertheless, as the smile is quite flat, also the purely Gaussian model seems to be a reasonable choice. 

As we want to have one model  with one single set of parameters, that captures the implied volatility features of all contracts, forthcoming and not, we can conclude that the two factor model performs overall very well and is a dignified winner of this empirical study.


%
%

\section{Conclusion}\label{sec:6}
In this paper, we present the theory of arbitrage possibilities when trading in futures and option contracts with overlapping delivery periods (see also \cite{BPV, borger, latini}).  In this setting, we discuss the necessary no-overlapping-arbitrage (NOA) conditions of a NIG-driven two factor model with Samuelson and seasonal factor. Our main purpose is the calibration of this model to power option prices via  Fourier transform methods in an overlapping-arbitrage free way.  This leads to an additional NOA restriction in our optimization problem, see Section \ref{sec:wholemarket}.  Looking at the market implied volatilities at our chosen trading day, we observe a pronounced forward-skewed volatility smile for the forthcoming delivery periods. The smile flattens out  when time to maturity increases. Thus, it is not surprising that our two factor model --- that accounts exactly for these short-term, long-term variations --- is able to fit the smile very well. 
\section*{Acknowledgments}
\noindent 
 T. Vargiolu acknowledges financial support from the research grant BIRD172407- 2017 of the University of Padova "New perspectives in stochastic methods for finance and energy markets".
This work has been initiated, and then revised, while M. Schmeck was visiting the Department of Mathematics of the University of Padova thanks to the fundings provided by the ``ACRI Young Investigator Training Program'' (YITP-EFI2018) and the YITP 2019 connected to the Energy Finance Italia Conference 2019 (UA.MB.D10-Dipartimento di Statistica e Metode Quantitativi, Universita' degli Studi di Milano - Bicocca), and while M. Piccirilli visited the University of Bielefeld.


\appendix

\section{Appendix}

\subsection{Inversion of the Fourier transform: Carr-Madan approach}

We here recall the Carr-Madan approach for recovering the option value from its Fourier transform (see Section \ref{sec:3}). If $\xi^{MT}_{t,T}$, defined in \eqref{xi_mt}, is integrable, by the Fourier inversion theorem, the modified time value of the option can be recovered by inverting again the last equation:
\begin{equation}\label{z_inv_mt}
z^{MT}_{t,T}(K)=\frac1{2\pi} \int_{-\infty}^{+\infty} e^{-iKv}\,\xi^{MT}_{t,T}(v)\,dv.
\end{equation}
Now, observe that, by the martingale property of $F(\cdot)$,
$$
(F(t)-K)\,\mathds{1}_{F(t)>K}=\bE\left[ (F(T)-K)\left| \mathcal{F}_t \right. \right]\,\mathds{1}_{F(t)>K}.
$$
Then, we can write from \eqref{density}
$$
z^{MT}_{t,T}(K)=\int_{-\infty}^{+\infty}(s-K)\,q_{t,T}(s)\,(\mathds{1}_{s>K}-\mathds{1}_{s(t)>K})\,ds
$$
where, for any $t\in[0,T]$, $s(t):=F(t)$ and $q_{t,T}$ is the density function of $F(T)$ conditioned up to time $t$. As a consequence, if the interchange of integrals holds (see Proposition \ref{change_int2}), \eqref{xi_mt} can be written as
\begin{align}\nonumber
\xi^{MT}_{t,T}(v)=&\int_{-\infty}^{+\infty} e^{ivK}\,\int_{-\infty}^{+\infty}(s-K)\,q_{t,T}(s)\,(\mathds{1}_{s>K}-\mathds{1}_{s(t)>K})\,ds\,dK\\ \label{change}
&=\int_{-\infty}^{+\infty} q_{t,T}(s)\, \int_{s(t)}^{s} e^{ivK}\,(s-K)\,dK\,ds\\\nonumber
&=\int_{-\infty}^{+\infty} q_{t,T}(s)\left\{ \left. \frac{e^{ivK}s}{iv}\right|_{K=s(t)}^{s}- \int_{s(t)}^{s} K\,e^{ivK}\,dK\right\} ds\\\nonumber
&=\int_{-\infty}^{+\infty} q_{t,T}(s)\left\{ -\frac{e^{ivs(t)}s}{iv}+\frac{e^{ivs(t)}s(t)}{iv}-\frac{e^{ivs}}{v^2}+\frac{e^{ivs(t)}}{v^2} \right\} ds\\\nonumber
&=-\frac1{v^2}\,\int_{-\infty}^{+\infty} q_{t,T}(s)\,e^{ivs} \, ds+\frac{e^{ivs(t)}}{v^2}-\frac{e^{ivs(t)}}{iv}\,\int_{-\infty}^{+\infty} q_{t,T}(s) (s-s(t)) \, ds.
\end{align}
Since $F$ is a martingale,
$$
\int_{-\infty}^{+\infty} q_{t,T}(s) (s-s(t)) \, ds = \bE\left[ (F(T)-F(t))\left| \mathcal{F}_t \right. \right]=0
$$
and, by definition of characteristic function (see also \eqref{cumulants}),
$$
\int_{-\infty}^{+\infty} q_{t,T}(s)\,e^{ivs} \, ds = \bE\left[e^{ivF(T)}\left| \mathcal{F}_t \right. \right] = e^{ivs(t)} {\Psi(t,T,v)},
$$
where $\Psi(t,T,v)$ is the characteristic function of $Z(t,T)$. 
Then, we finally arrive to
\begin{equation}\label{xi_analytic_mt}
\xi^{MT}_{t,T}(v)=e^{ivF(t)}\frac{1-{\Psi(t,T,v)}}{v^2}.
\end{equation}
In analogy to the modified option approach, we have that $\xi^{MT}_{t,T}$ has even real part and odd imaginary part and so 
\begin{equation}\label{z_inv_re2}
z^{MT}_{t,T}(K)=\frac1{\pi} \int_{0}^{+\infty} \mathrm{Re}\left( e^{-iKv}\,\xi^{MT}_{t,T}(v)\right)\,dv.
\end{equation}
Finally, by recalling \eqref{mod_time}, the option value is computed from \eqref{z_inv_re2} by
$$
C(t;T,K)=z^{MT}_{t,T}(K)+(F(t)-K)_+.
$$

In order to apply this formula, we need to justify the interchange of integrals operated in \eqref{change} under the following integrability assumptions on the futures price process.
\begin{prop}\label{change_int2}
If $Z(T,T_1,T_2)$ is square-integrable, then 
\begin{align*}
\xi_{t,T}(v)=&\int_{-\infty}^{+\infty} e^{ivK}\,\int_{-\infty}^{+\infty}(s-K)\,q_{t,T}(s)\,(\mathds{1}_{s>K}-\mathds{1}_{s(t)>K})\,ds\,dK\\
&=\int_{-\infty}^{+\infty} q_{t,T}(s)\, \int_{s(t)}^{s} e^{ivK}\,(s-K)\,dK\,ds.
\end{align*}
\end{prop}
\begin{proof}
It is enough to show that the integral of the absolute value of the integrand with respect to $K$, i.e.  
$$
\int_{-\infty}^{+\infty} |s-K|\,|\mathds{1}_{s>K}-\mathds{1}_{s(t)>K}|\,dK
$$
can be integrated in $s$ against the density $q_{t,T}(s)$. Since this is equal to
\begin{align*}
\int_{s(t)}^{s} (s-K)\,dK &= \frac{1}{2}(s-s(t))^2,
\end{align*}
the integrability with respect to the density is equivalent to the existence of the second moment of $Z(t,T,T_1,T_2)$.
\end{proof}

\subsection{Characteristic function of the two-factor NIG model}

As already stated in Section 4.2 for general multifactor additive models, the characteristic function of two-factor NIG model is defined as a function of $v\in\bR$ by
$$
\Psi(t,T,T_1,T_2,v) = \bE\left[ e^{ivZ(t,T,T_1,T_2)} \left|\mathcal{F}_t\right.\right],
$$
where  $Z(t,T,T_1,T_2) := F(T,T_1,T_2)-F(t,T_1,T_2)$, and it can be shown to be equal to
\begin{align}\label{cum_app}
\log\Psi(t,T,T_1,T_2,v) &= \psi_1(t,T;y\Gamma_1(\cdot,T_1,T_2)) +\psi_2(t,T;y\Gamma_2(T_1,T_2)).
\end{align}
 
First, we need to compute the cumulant function $\widetilde\psi_j(\theta)$ of $J_j$ $(j=1,2)$, that can be computed from the corresponding cumulant $\widetilde\psi_{L_j}(\theta)$ of $L_j$ as follows. From the definition of $J_j$ we know that
$$
\widetilde\psi_j(\theta)=\widetilde\psi_{L_j}(\theta)-i\theta \left(\chi_j+\int_{|y|\geq 1} y\,\nu_j(dy)\right),
$$
where 
$$
\widetilde\psi_{L_j}(\theta)=\delta\left(\sqrt{\alpha_j^2-\beta_j^2}-\sqrt{\alpha_j^2-(\beta_j+i\theta)^2}\right) + i\theta m_j.
$$
Moreover, since $L_j(1)$ is a NIG distributed random variable, its expected value is
$$
\bE[L_j(1)]=\left(\chi_j+\int_{|y|\geq 1} y\,\nu_j(dy)\right)=m_j+\frac{\delta_j\beta_j}{\sqrt{\alpha_j^2-\beta_j^2}},
$$
so that, by replacing it in the expression above, we find that
\begin{equation}\label{ch_fun_center}
\widetilde\psi_j(\theta)=\delta_j\,\left(\sqrt{\alpha_j^2-\beta_j^2}-\sqrt{\alpha_j^2-(\beta_j+i\theta)^2}  -i\theta\,\frac{\beta_j}{\sqrt{\alpha_j^2-\beta_j^2}}\right).
\end{equation}
In particular, we observe that $J_j(1)$ is a NIG distributed random variable with parameters $(\alpha_j,\beta_j,\delta_j,-\frac{\delta_j\beta_j}{\sqrt{\alpha_j^2-\beta_j^2}})$. Now, we can compute the two jump components of the cumulant function in \eqref{cum_app}, 
by inserting the corresponding expressions of $\widetilde\psi_j(\theta)$ from \eqref{ch_fun_center} and $\Gamma_j$ as in \eqref{gamma1_3}--\eqref{gamma2_3} for $j=1,2$. The second coefficient can be directly computed as 
\begin{align}\nonumber
\psi_2(t,T;y\Gamma_2(T_1,T_2))&=(T-t) \,\delta_2 \Bigl(\sqrt{\alpha_2^2-\beta_2^2}-\sqrt{\alpha_2^2-(\beta_2+i y\Gamma_2(T_1,T_2))^2}\\
\label{psi2}
&-i y\Gamma_2(T_1,T_2)\frac{\beta_2}{\sqrt{\alpha_2^2-\beta_2^2}}\Bigr),
\end{align}
while the first requires an integration in time:
\begin{align}\nonumber
\psi_1(t,T;y\Gamma_1(\cdot,T_1,T_2))& = (T-t) \,\delta_1\,\sqrt{\alpha_1^2-\beta_1^2}\\  \nonumber 
&-\delta_1\,\int_t^T\sqrt{\alpha_1^2-(\beta_1+iy\Gamma_1(u,T_1,T_2))^2}\,du\\
\label{psi1_raw}
&-iy\frac{\delta_1\beta_1}{\sqrt{\alpha_1^2-\beta_1^2}} \int_t^T\Gamma_1(u,T_1,T_2)\,du.
\end{align}
Let us recall that, if $\zeta := \Gamma_2(T_1,T_2)$ is positive, by the properties of the NIG distribution (see e.g. \cite{BN}), $\zeta\cdot J_2(1)$ is a NIG distributed random variable with parameters $(\alpha_2/\zeta,\beta_2/\zeta,\delta_2\zeta,-\frac{\delta_2\zeta\beta_2}{\sqrt{\alpha_2^2-\beta_2^2}})$. Consequently, it is easy to see that we can assume without loss of generality that $\delta_2=1$, so that 
\begin{align}\nonumber
\psi_2(t,T;y\Gamma_2(T_1,T_2))&=(T-t) \,\Bigl(\sqrt{\alpha_2^2-\beta_2^2}-\sqrt{\alpha_2^2-(\beta_2+i y\Gamma_2(T_1,T_2))^2}\\
\label{}
&-i y\Gamma_2(T_1,T_2)\frac{\beta_2}{\sqrt{\alpha_2^2-\beta_2^2}}\Bigr),
\end{align}
For the same reason we can assume that $\delta_1=1$. By replacing $\Gamma_1(u,T_1,T_2) = e^{\mu u}\,\frac{\gamma_1(e^{-\mu T_1}-e^{-\mu T_2})}{\mu (T_2-T_1)}:=e^{\mu u}\,\widetilde\Gamma_1(T_1,T_2)$ in \eqref{psi1_raw} and integrating, we get
\begin{align}\nonumber
\psi_1(t,T;y\Gamma_1(\cdot,T_1,T_2))&=(T-t)\,\sqrt{\alpha_1^2-\beta_1^2}-\eta(c(y,T))+\eta(c(y,t))\\ 
\label{psi1}
&-\frac{iy\,\widetilde\Gamma_1(T_1,T_2)\,\beta_1 (e^{\mu T}-e^{\mu t})}{\mu\sqrt{\alpha_1^2-\beta_1^2}}, 
\end{align}
where
\begin{align}\nonumber
\eta(w):=&\frac{1}\mu\left( \sqrt{\alpha_1^2+w^2}-i\beta_1\arcsinh{\frac{w}{\alpha_1}} \right.\\
\label{eta}
&\left.-\sqrt{\alpha_1^2-\beta_1^2}\, \log \frac{2\alpha_1^2\left( \alpha_1^2-i\beta_1 w+\sqrt{\alpha_1^2-\beta_1^2}\sqrt{\alpha_1^2+w^2} \right)}
{(w+i\beta_1) (\alpha_1^2-\beta_1^2)^\frac32} \right),\\
\nonumber \\ \label{c3}
c(v,u)&:=v\,\widetilde\Gamma_1(T_1,T_2)\,e^{\mu u}-i\beta_1.
\end{align}
Finally, the cumulant function of $Z$ is explicitly given by
\begin{align*}
\log\Psi(t,T,T_1,T_2,y) &=\psi_1(t,T;y\Gamma_1(\cdot,T_1,T_2)) +\psi_2(t,T;y\Gamma_2(T_1,T_2))\\
&=(T-t) \,\Biggl\{\sqrt{\alpha_1^2-\beta_1^2}+\sqrt{\alpha_2^2-\beta_2^2}\\
&-iy\,\Biggl(\frac{\widetilde\Gamma_1(T_1,T_2)\beta_1}{\mu\sqrt{\alpha_1^2-\beta_1^2}}\,\frac{ (e^{\mu T}-e^{\mu t})}{T-t}
+\frac{\Gamma_2(T_1,T_2)\,\beta_2}{\sqrt{\alpha_2^2-\beta_2^2}}\Biggr)\\
&-\frac{\eta(c(y,T))-\eta(c(y,t))}{T-t}-\sqrt{\alpha_2^2-(\beta_2+i y\Gamma_2(T_1,T_2))^2}\Biggr\}.
\end{align*}
with $\eta(\cdot)$ and $c(\cdot,\cdot)$ as in \eqref{eta}--\eqref{c3}.

\section{Figures}\label{app:figures}

\begin{center}
\begin{figure} 
	
	\begin{subfigure}{0.7\textwidth}
		\includegraphics[width=\linewidth]{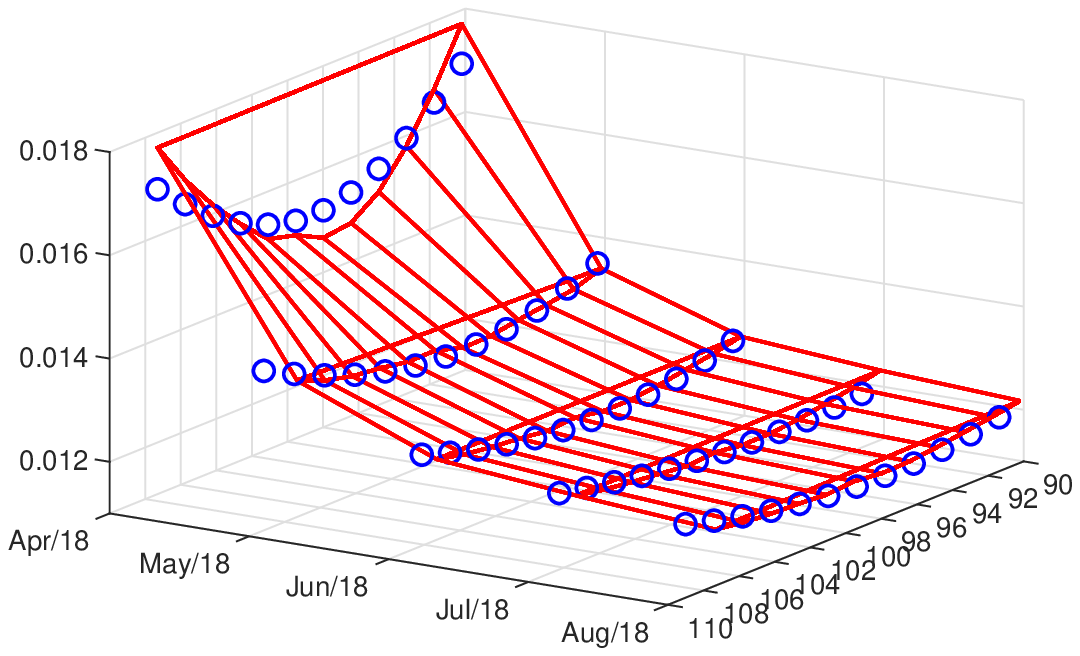}
	\end{subfigure}

	\begin{subfigure}{0.7\textwidth}
		\includegraphics[width=\linewidth]{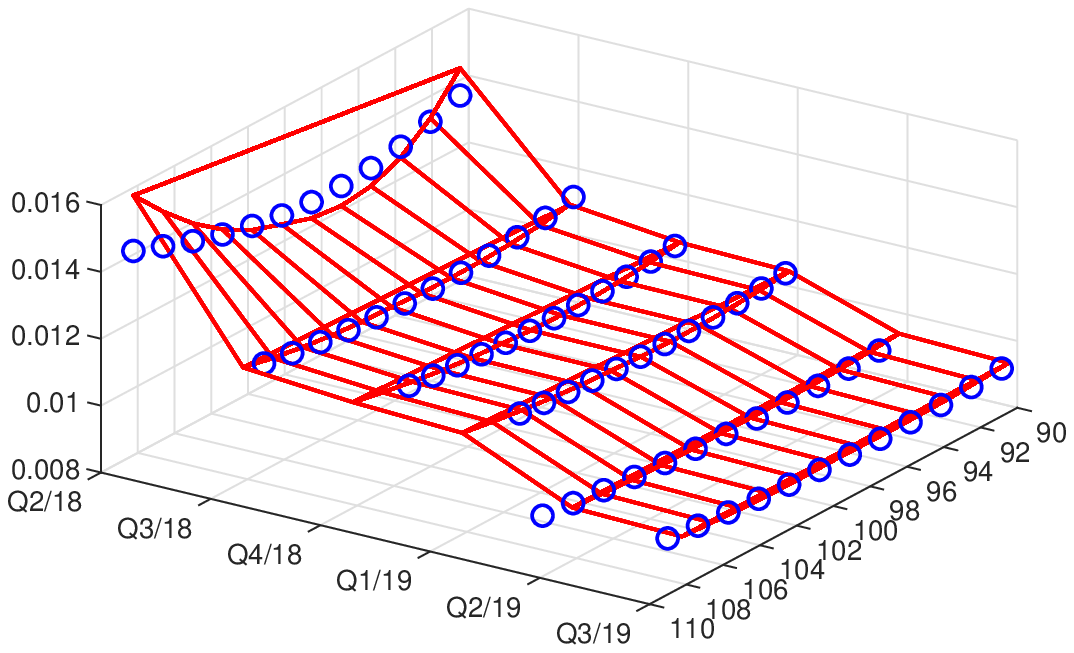}
	\end{subfigure}\hspace*{\fill}

	\begin{subfigure}{0.7\textwidth}
		\includegraphics[width=\linewidth]{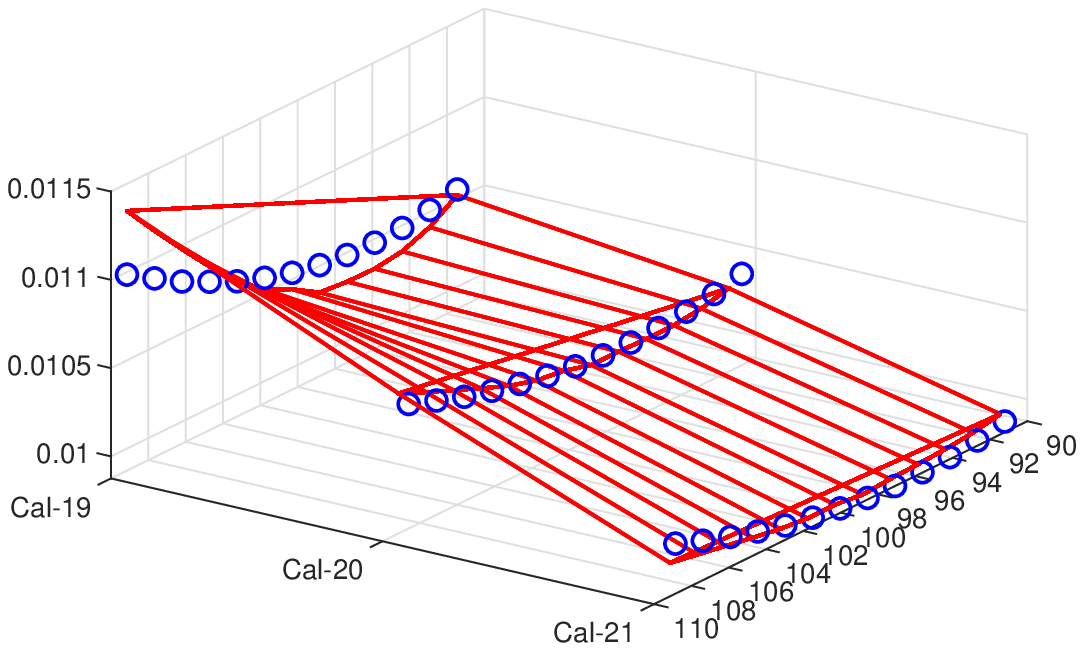}
	\end{subfigure}\hspace*{\fill}

	\caption{Surface plot of implied volatilites for montly, quarterly and yearly delivery periods. Empirical values (dotted, blue) and two factor model (red).} \label{fig:surface}
\end{figure}
\end{center}

\begin{figure} 
	\begin{subfigure}{0.5\textwidth}
		\includegraphics[width=\linewidth]{apr18}
	\end{subfigure}\hspace*{\fill}
	\begin{subfigure}{0.5\textwidth}
		\includegraphics[width=\linewidth]{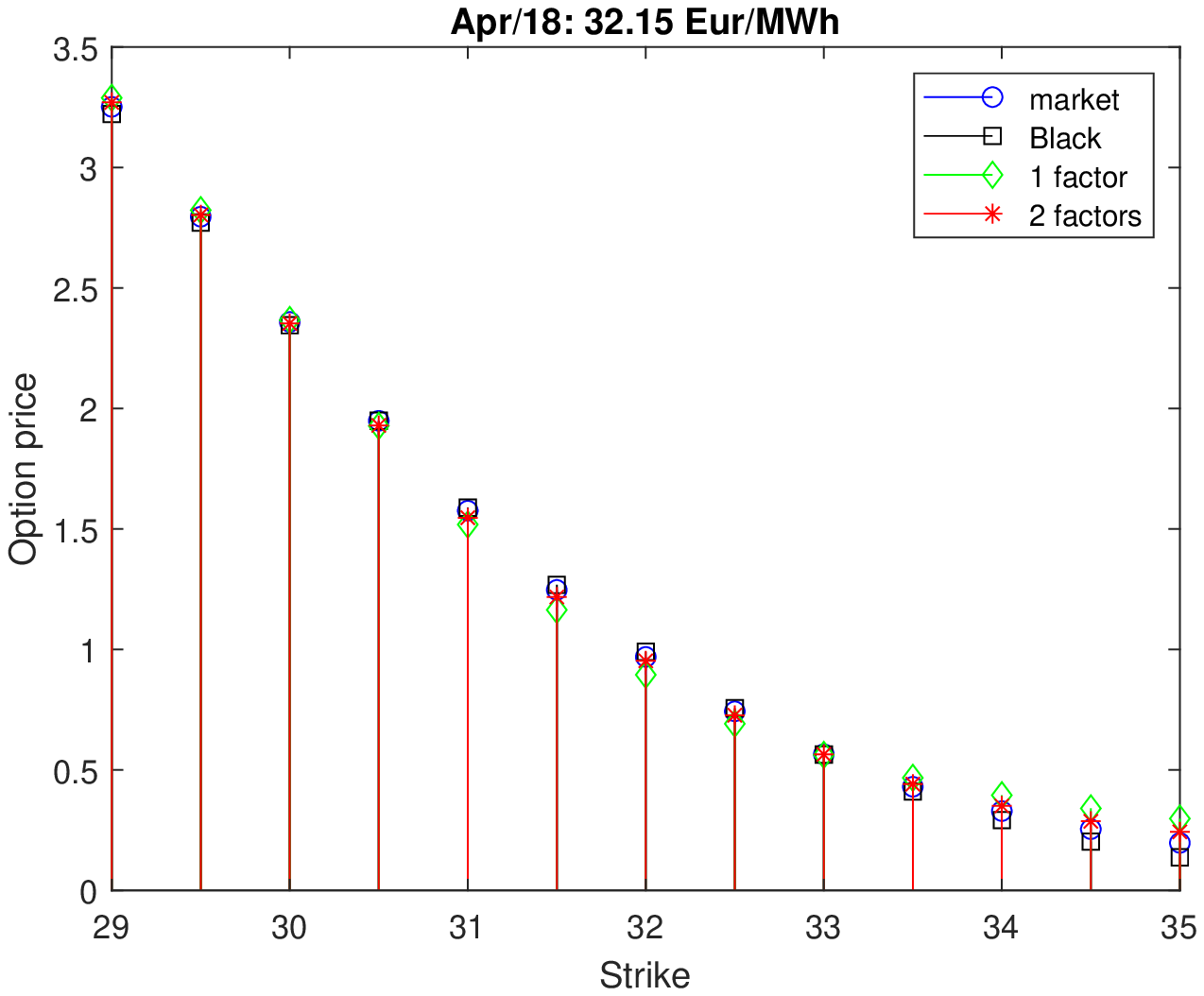}
	\end{subfigure}
	
	\vspace{30pt}
	\medskip
	\begin{subfigure}{0.5\textwidth}
		\includegraphics[width=\linewidth]{may18}
	\end{subfigure}\hspace*{\fill}
	\begin{subfigure}{0.5\textwidth}
		\includegraphics[width=\linewidth]{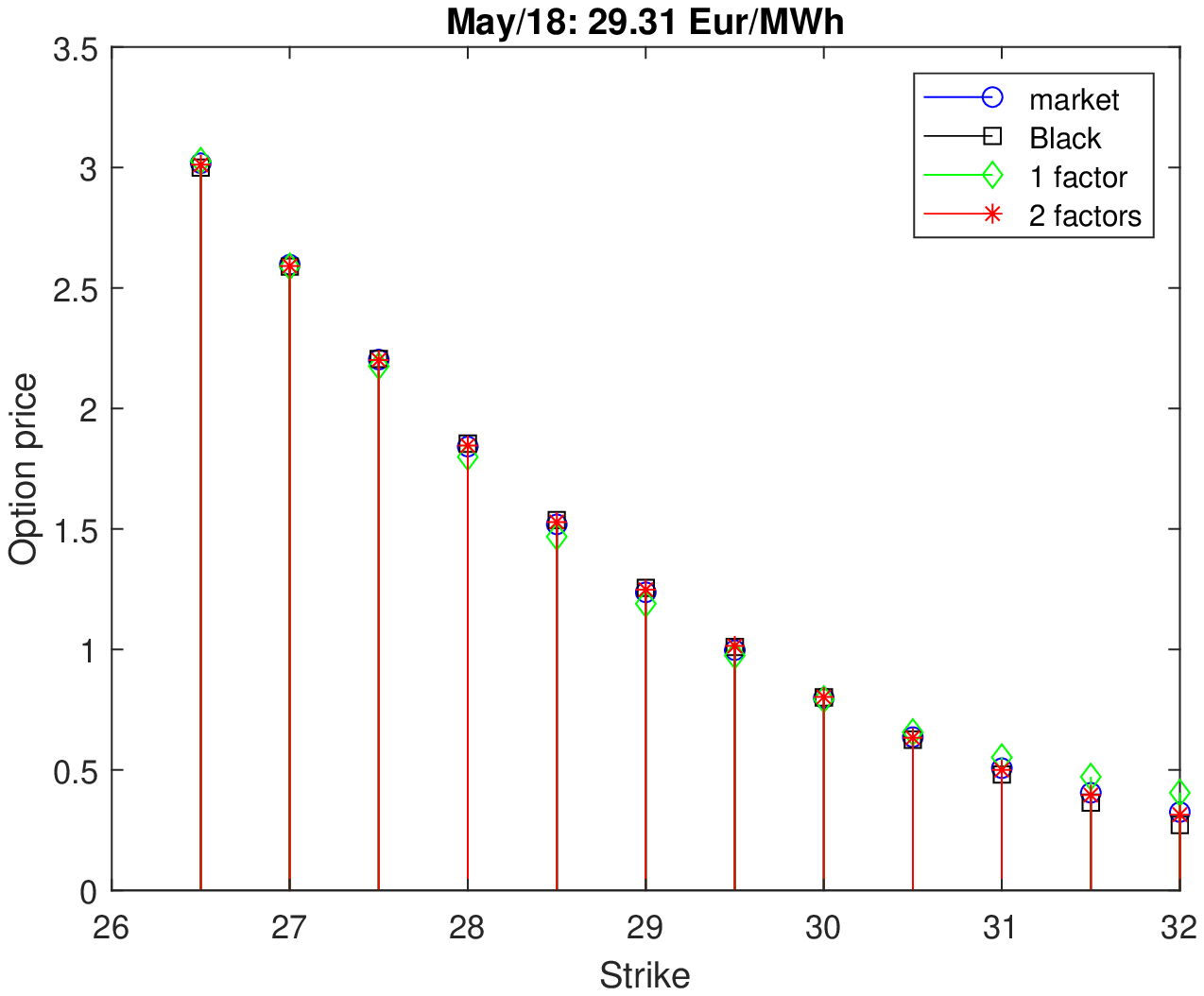}
	\end{subfigure}
	
	\vspace{30pt}
	\medskip
	
	\begin{subfigure}{0.5\textwidth}
		\includegraphics[width=\linewidth]{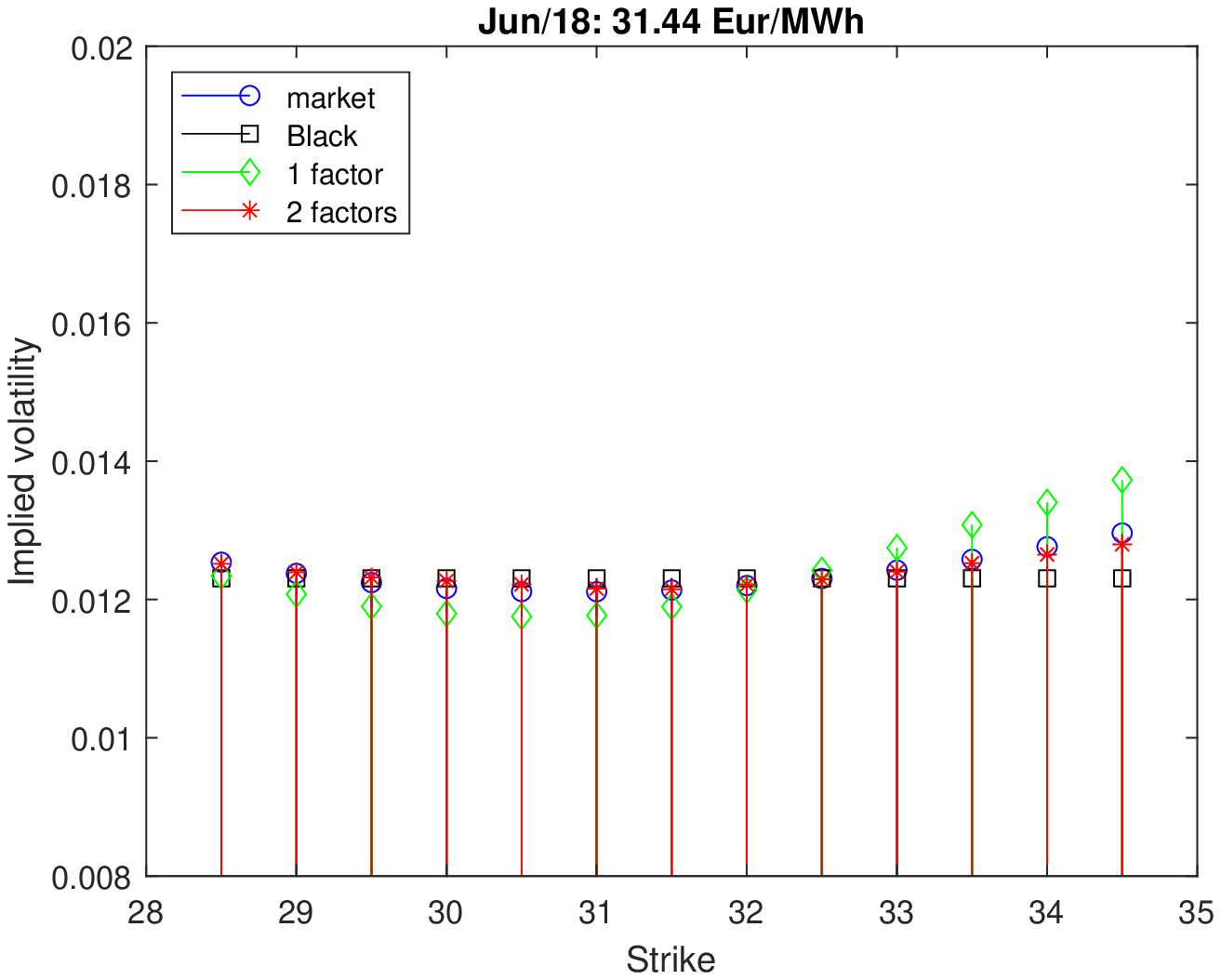}
	\end{subfigure}\hspace*{\fill}
	\begin{subfigure}{0.5\textwidth}
		\includegraphics[width=\linewidth]{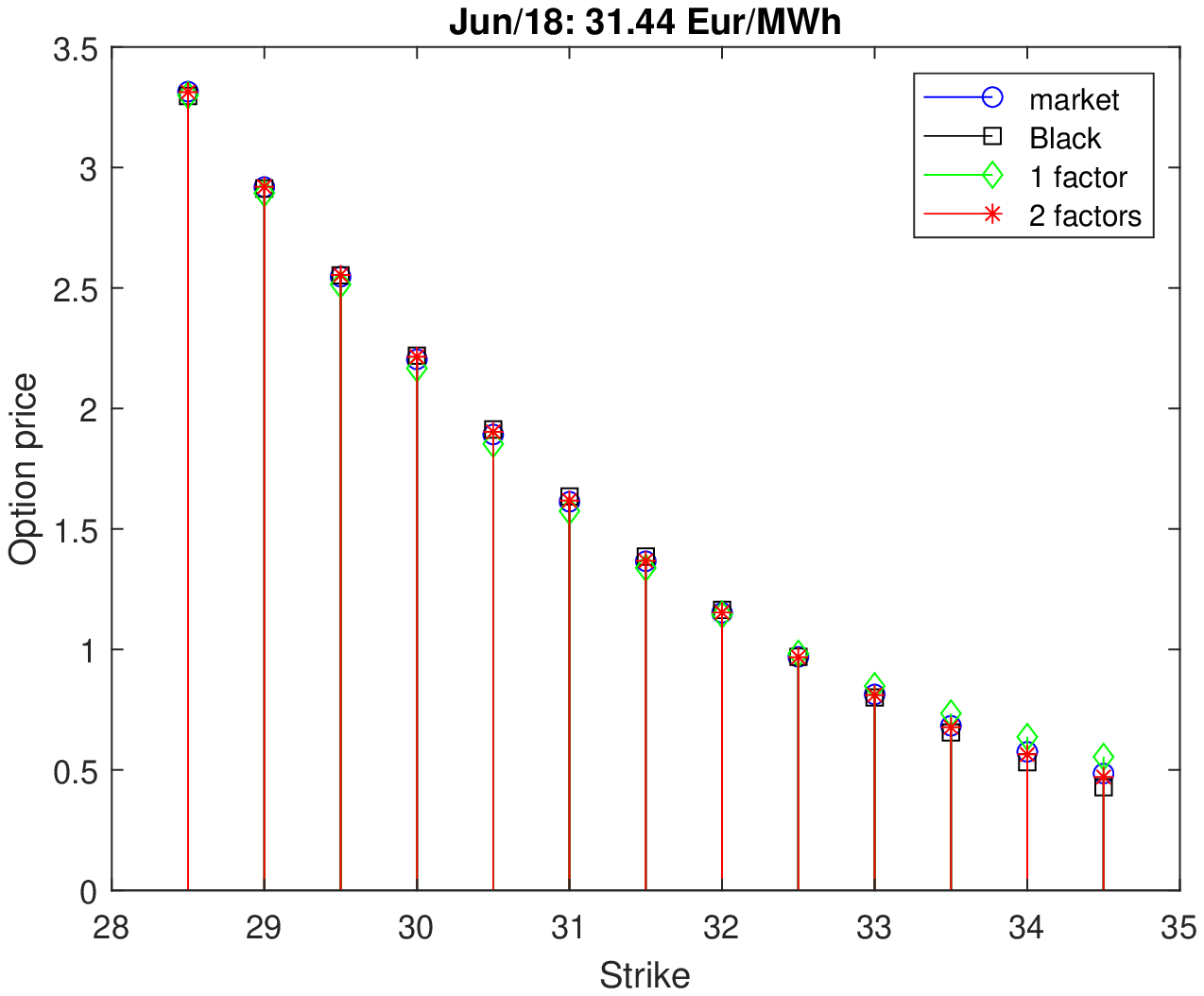}
	\end{subfigure}
	
%
	\vspace{30pt}
	
	
	\caption{Implied volatility for the Black, one-factor and two-factor model compared to the empirical implied volatilities of all the options listed at March 5, 2018 (left). The corresponding underlying current price is indicated above each plot. On the right hand side the corresponding prices are shown. Monthly delivery periods from April 18 - June 18.} \label{fig:1_ivm1}
\end{figure}

\begin{figure} 

	\begin{subfigure}{0.5\textwidth}
		\includegraphics[width=\linewidth]{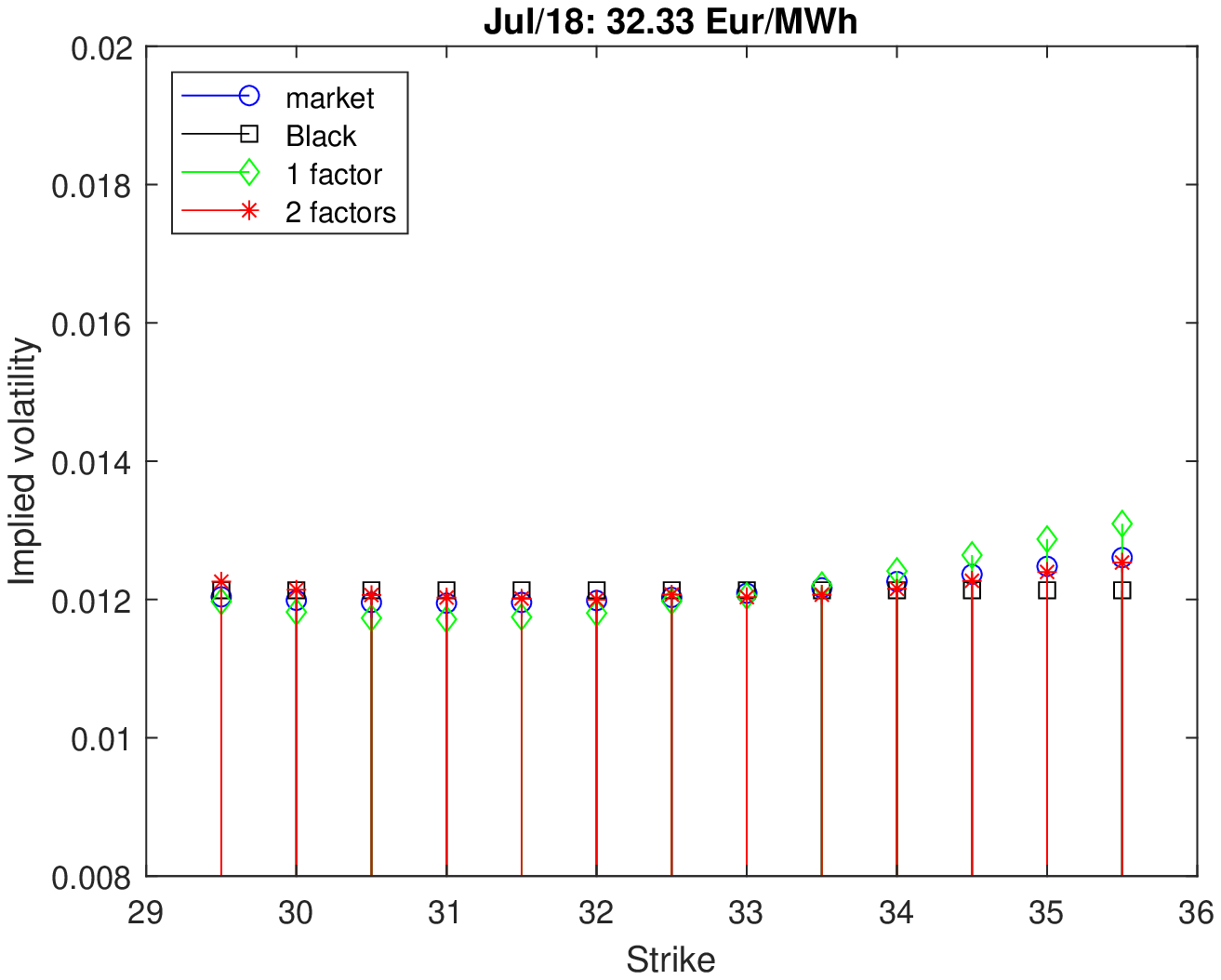}
	\end{subfigure}\hspace*{\fill}
	\begin{subfigure}{0.5\textwidth}
		\includegraphics[width=\linewidth]{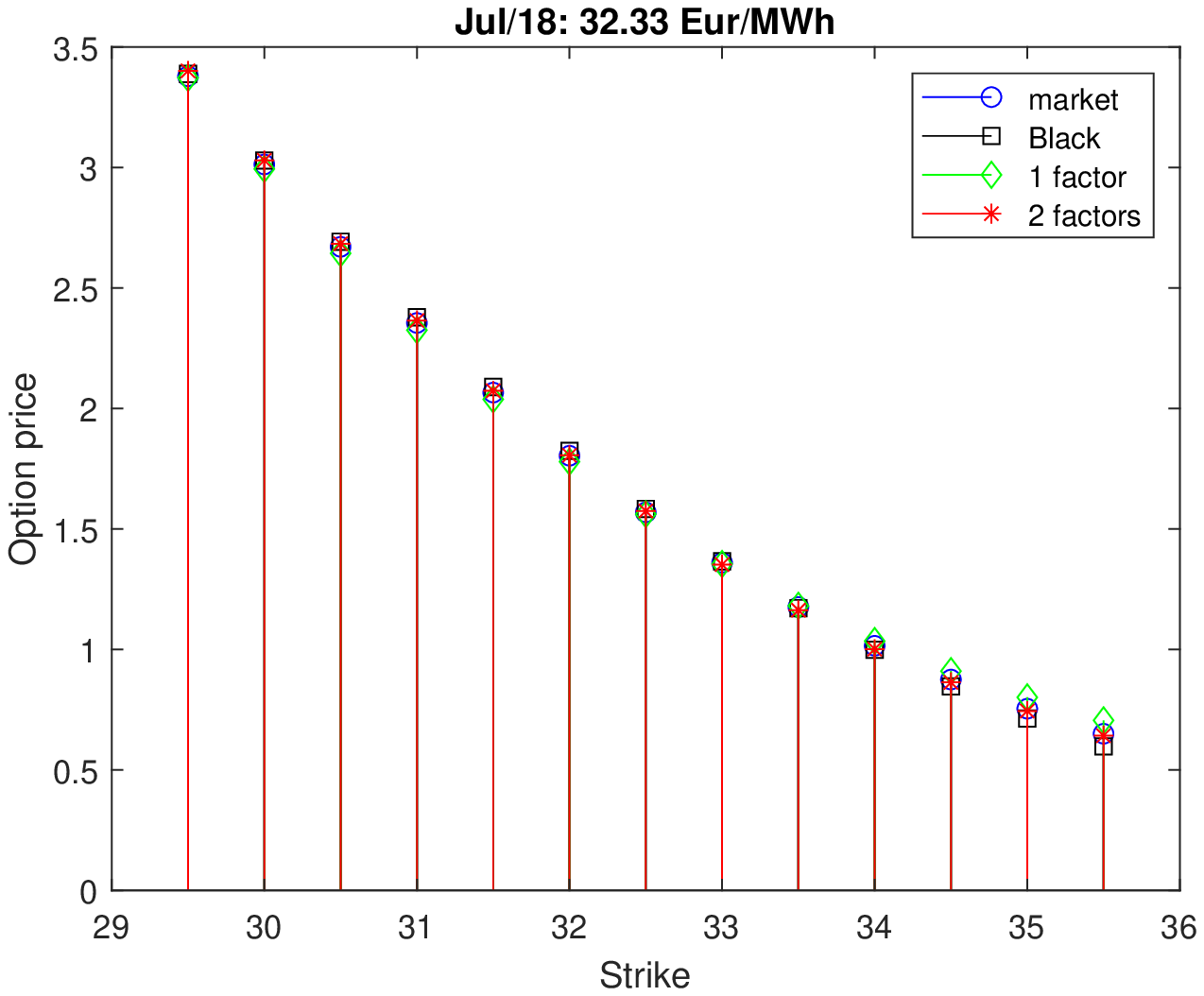}
	\end{subfigure}
	
		\vspace{30pt}
		\medskip

	\begin{subfigure}{0.5\textwidth}
		\includegraphics[width=\linewidth]{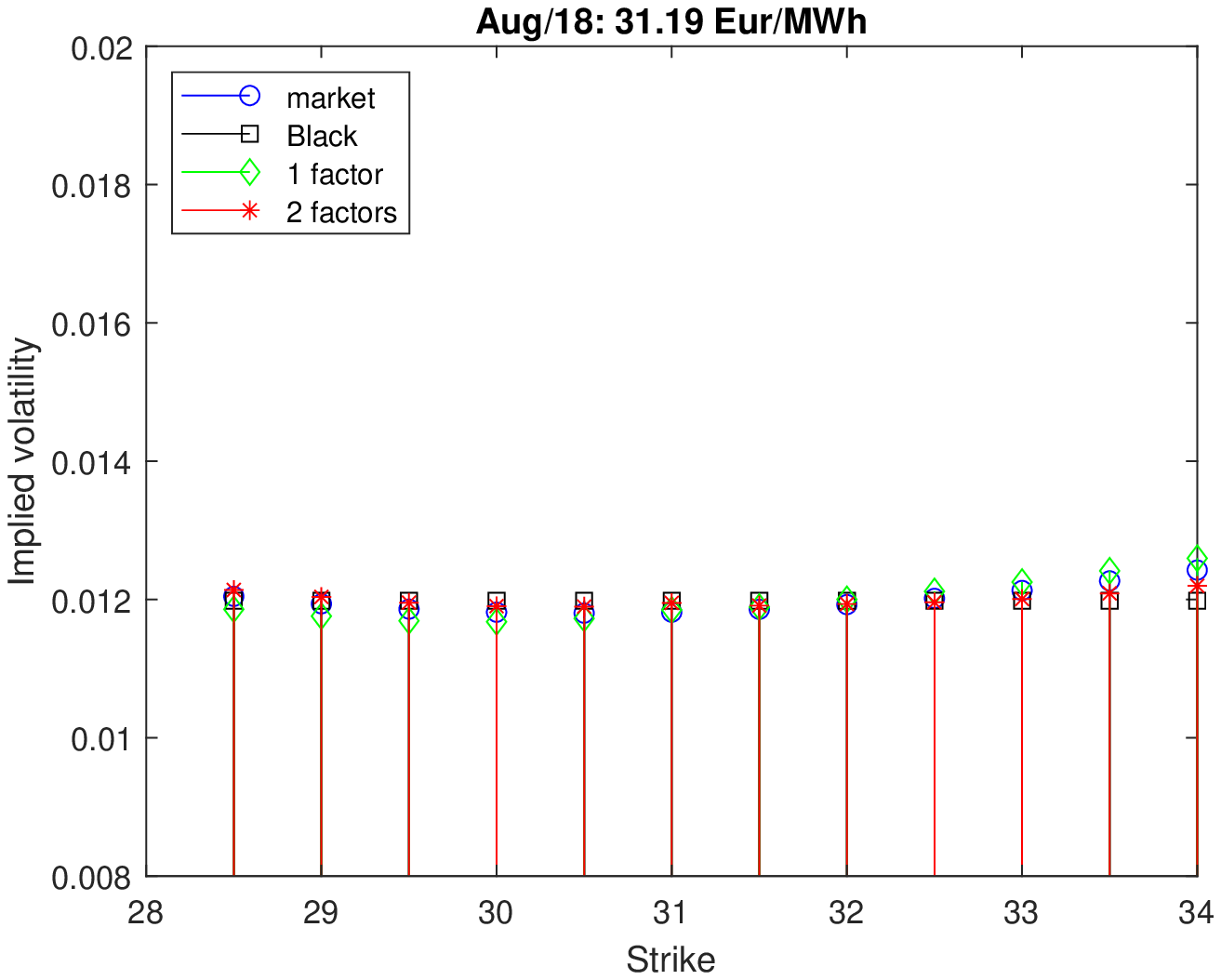}
	\end{subfigure}\hspace*{\fill}
	\begin{subfigure}{0.5\textwidth}
		\includegraphics[width=\linewidth]{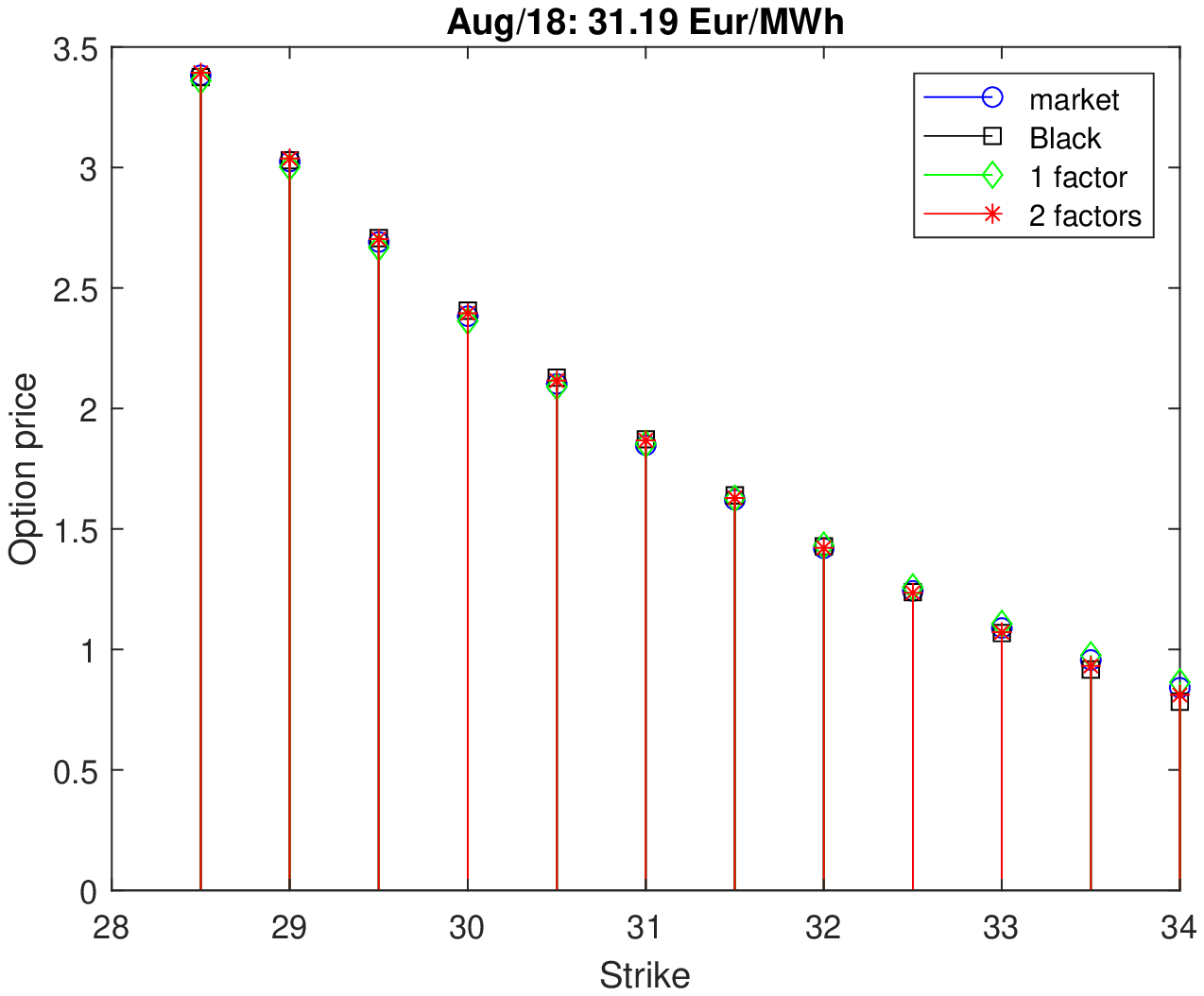}
	\end{subfigure}
	
	
	\vspace{30pt}
	
	\caption{Implied volatility for the Black, one-factor and two-factor model compared to the empirical implied volatilities of all the options listed at March 5, 2018 (left). The corresponding underlying current price is indicated above each plot. On the right hand side the corresponding prices are shown. Monthly delivery periods, July and August 2018.} \label{fig:1_ivm2}
\end{figure}

\begin{figure} 
	\begin{subfigure}{0.5\textwidth}
		\includegraphics[width=\linewidth]{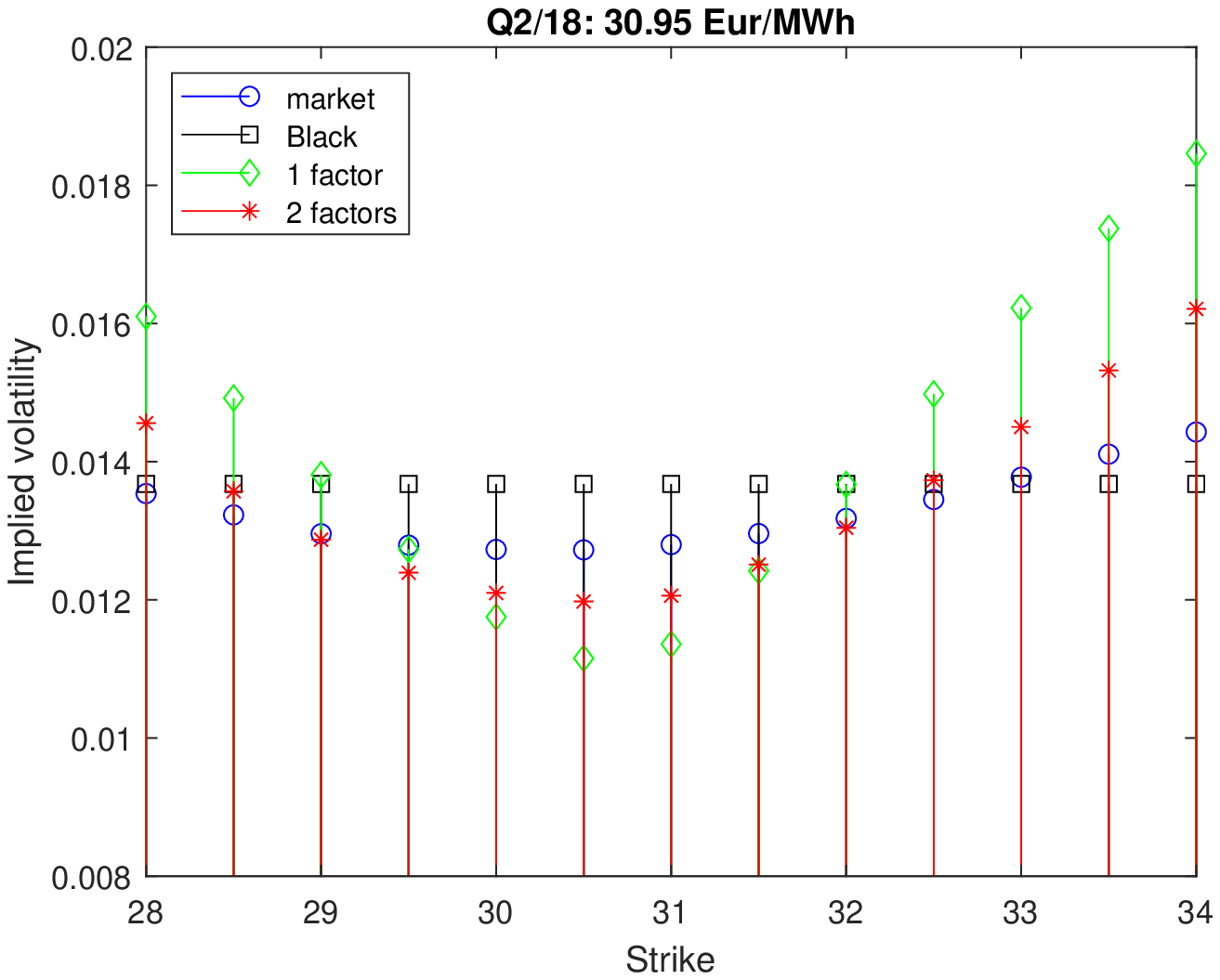}
	\end{subfigure}\hspace*{\fill}
	\begin{subfigure}{0.5\textwidth}
		\includegraphics[width=\linewidth]{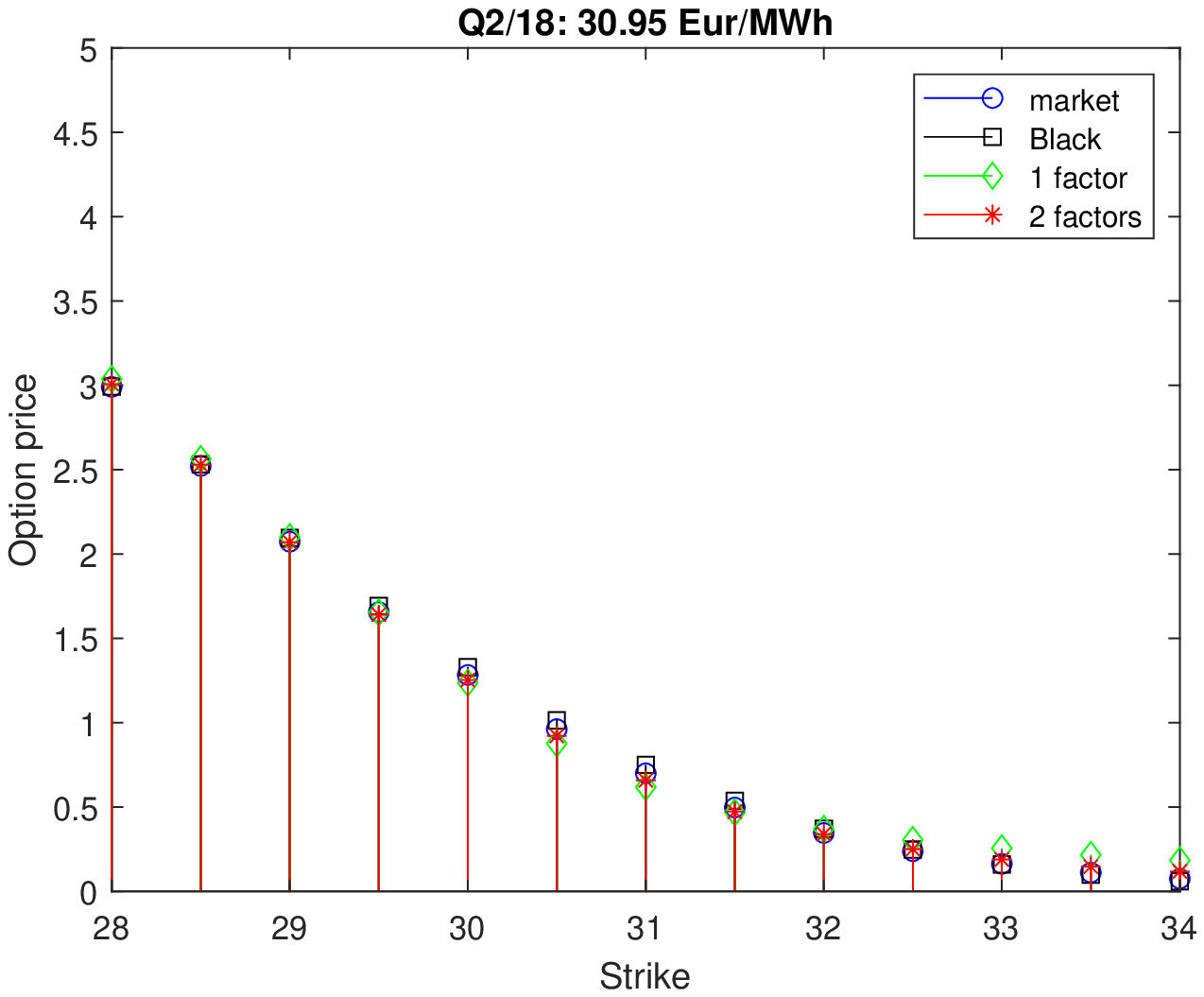}
	\end{subfigure}
	
	\vspace{30pt}
	\medskip
	\begin{subfigure}{0.5\textwidth}
		\includegraphics[width=\linewidth]{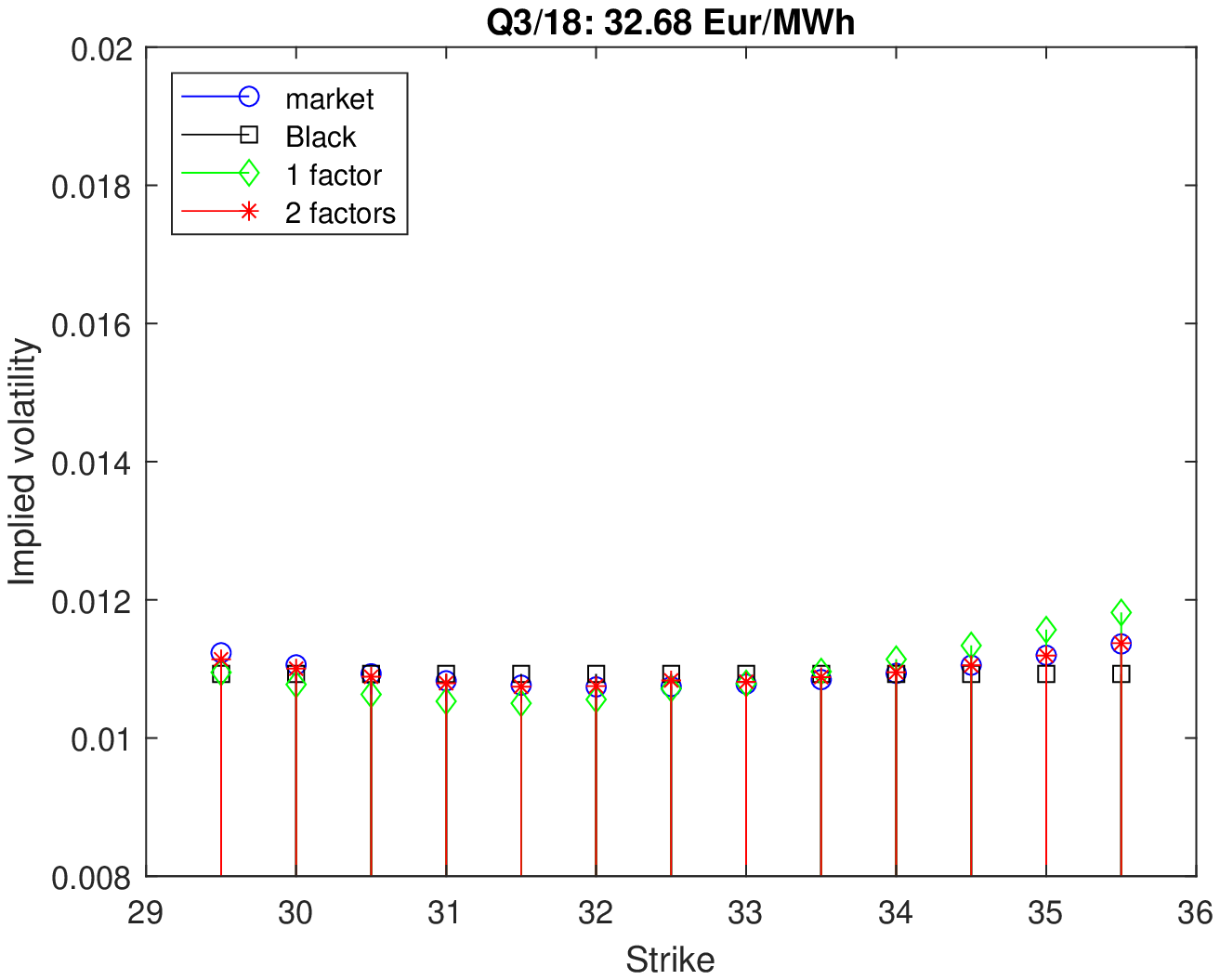}
	\end{subfigure}\hspace*{\fill}
	\begin{subfigure}{0.5\textwidth}
		\includegraphics[width=\linewidth]{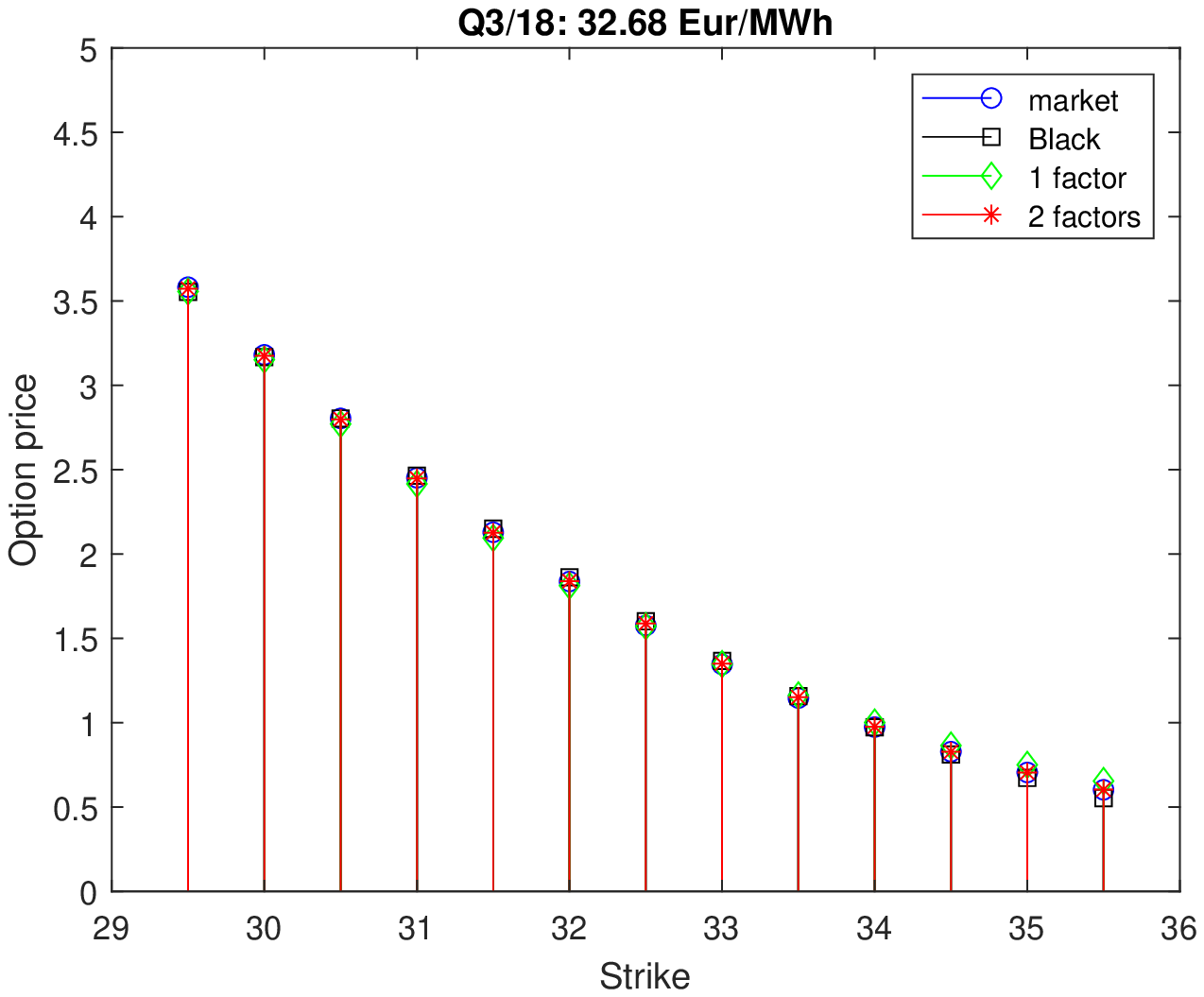}
	\end{subfigure}
	
	\vspace{30pt}
	\medskip
	
	\begin{subfigure}{0.5\textwidth}
		\includegraphics[width=\linewidth]{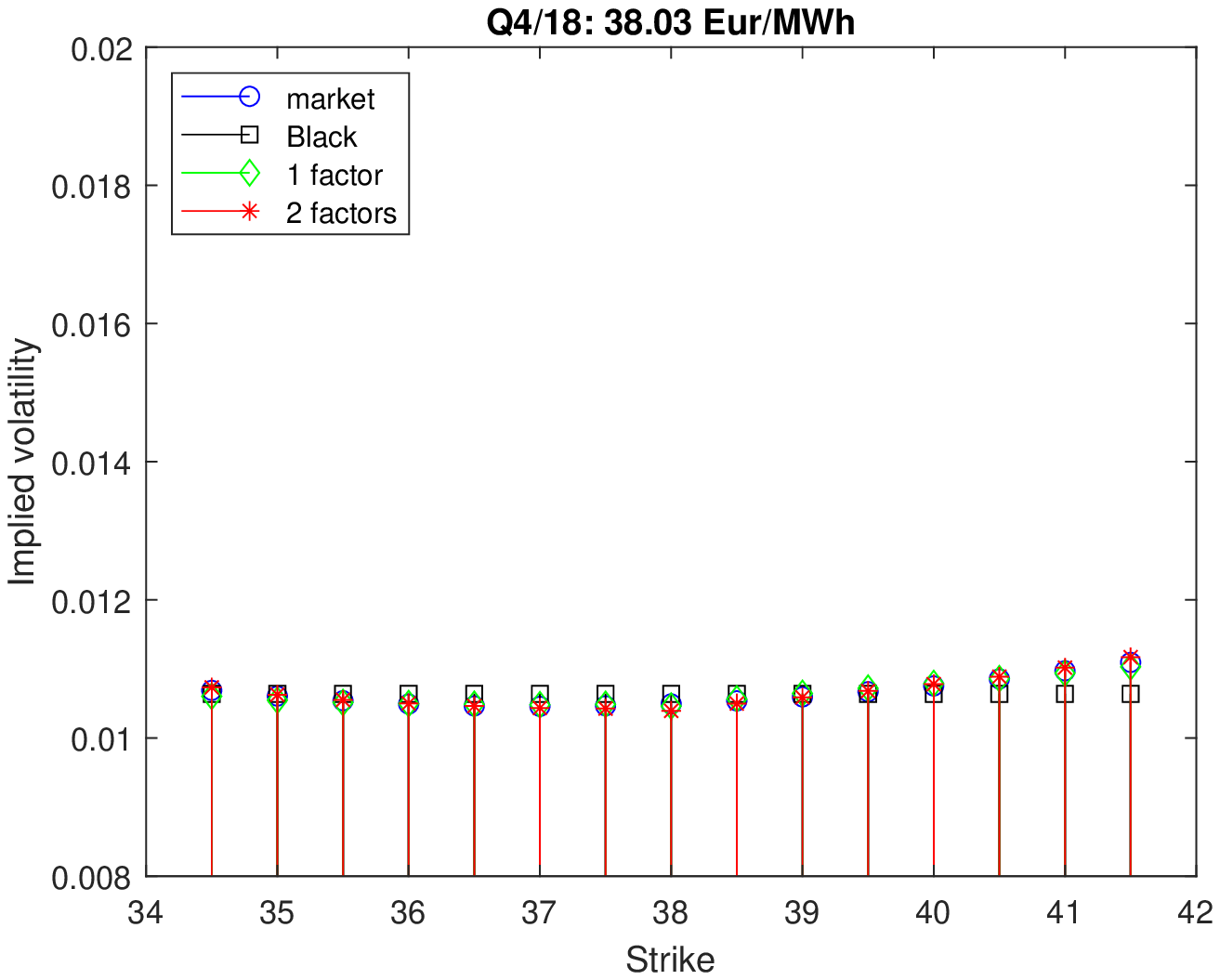}
	\end{subfigure}\hspace*{\fill}
	\begin{subfigure}{0.5\textwidth}
		\includegraphics[width=\linewidth]{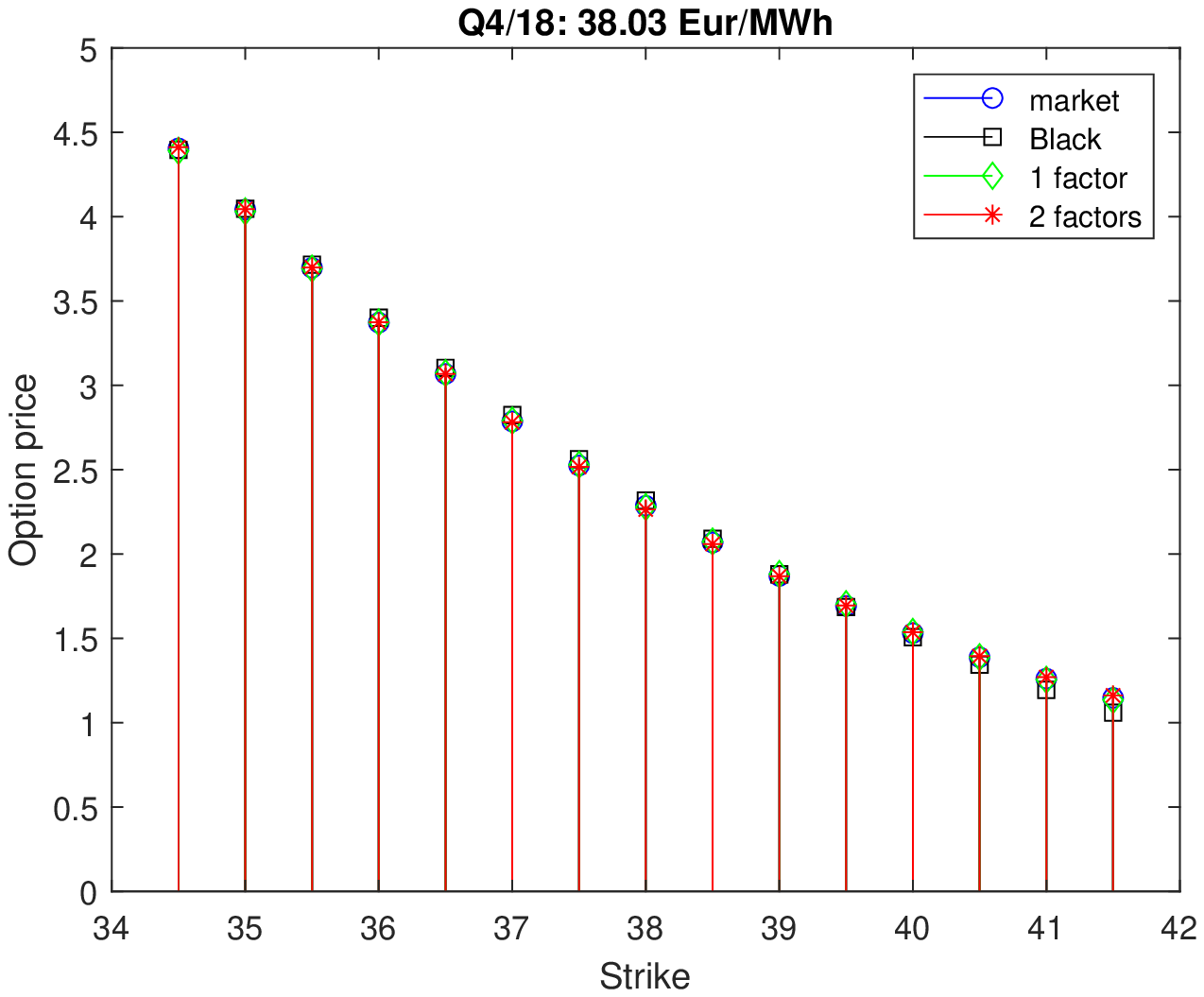}
	\end{subfigure}

	\vspace{30pt}
	
	
	\caption{Implied volatility for the Black, one-factor and two-factor model compared to the empirical implied volatilities of all the options listed at March 5, 2018 (left). The corresponding underlying current price is indicated above each plot. On the right hand side the corresponding prices are shown. Quarterly delivery periods in 2018.} \label{fig:2_ivq1}
\end{figure}

\begin{figure} 

	\vspace{30pt}
	\medskip
	\begin{subfigure}{0.5\textwidth}
		\includegraphics[width=\linewidth]{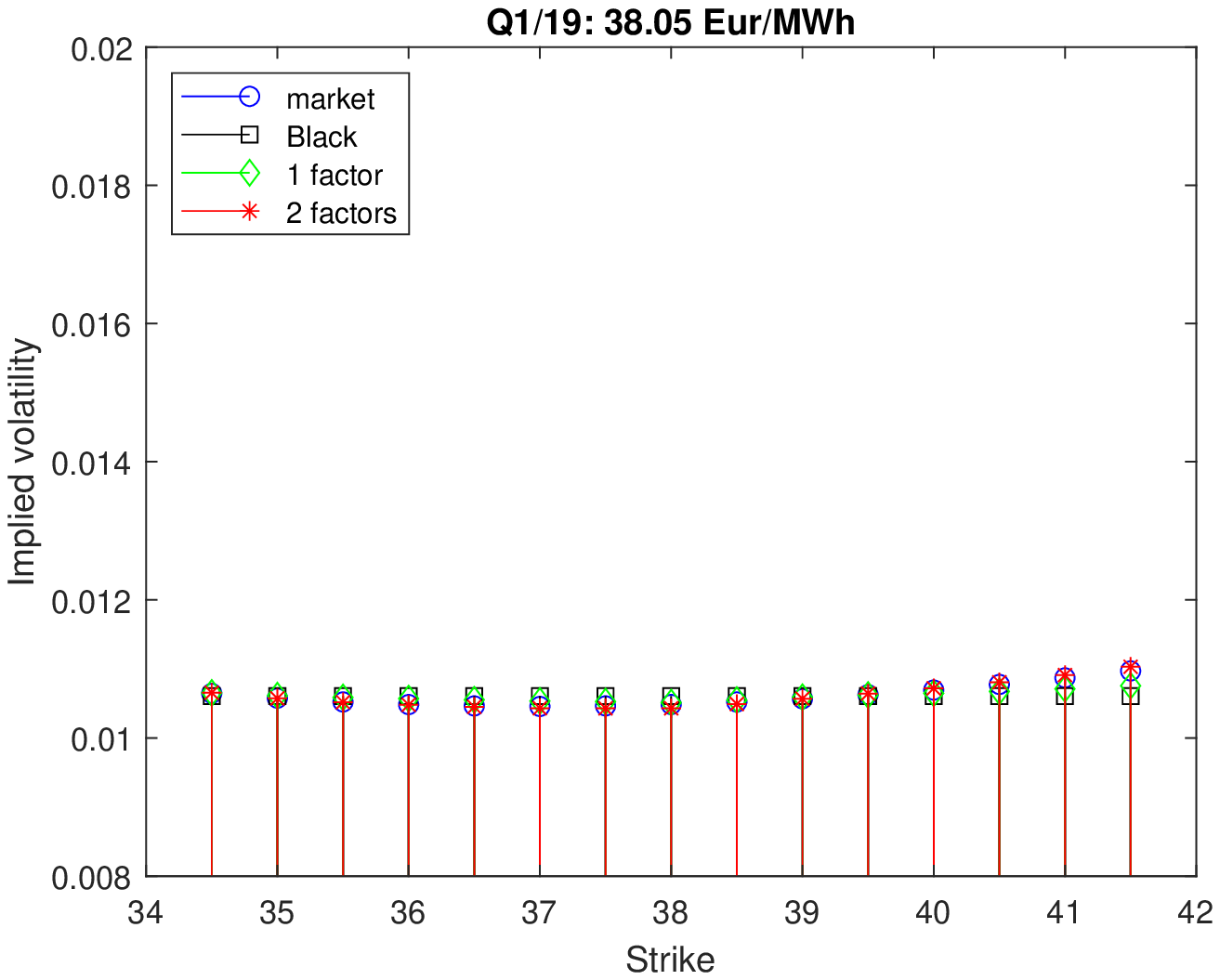}
	\end{subfigure}\hspace*{\fill}
	\begin{subfigure}{0.5\textwidth}
		\includegraphics[width=\linewidth]{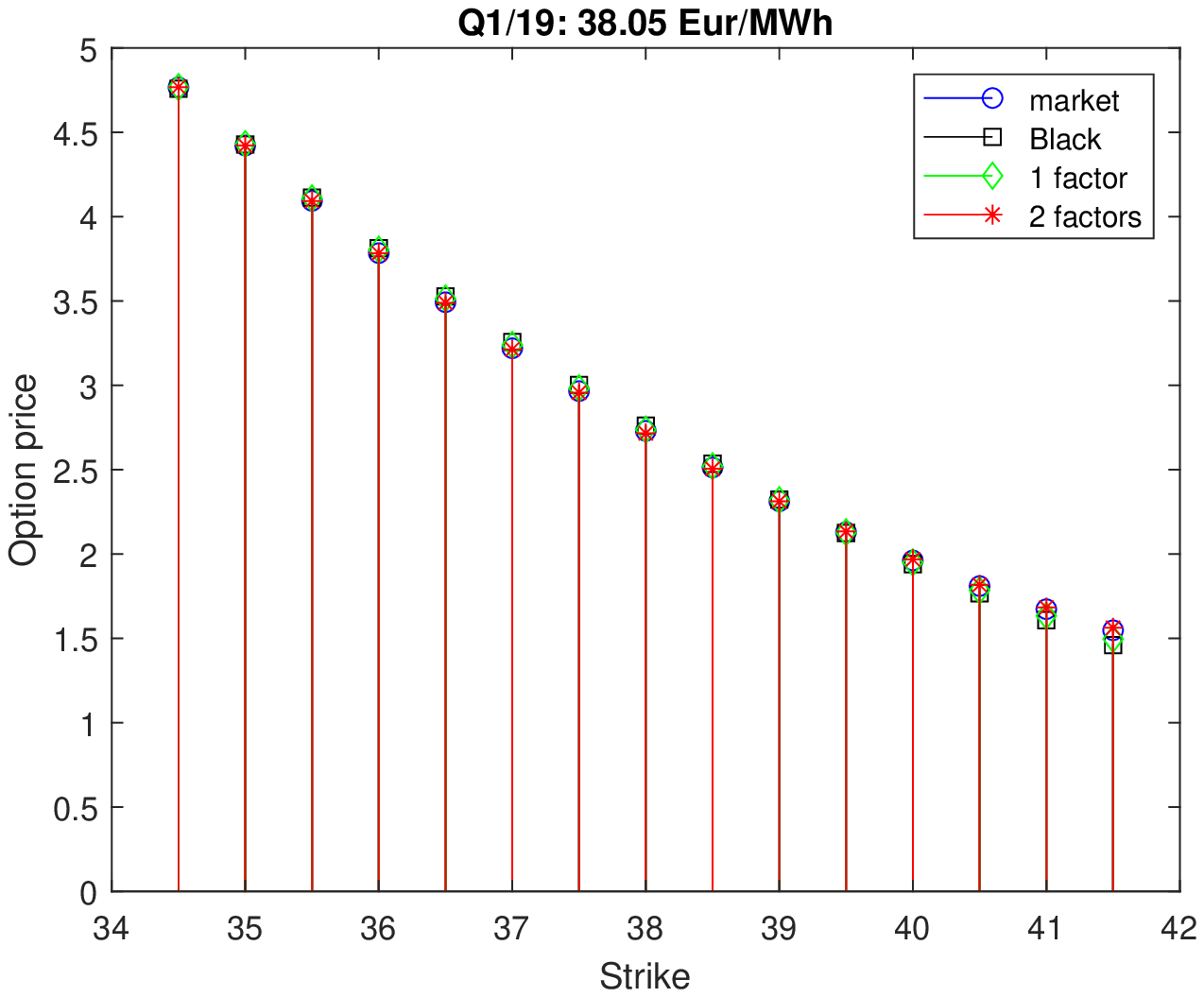}
	\end{subfigure}
	
	\vspace{30pt}
	\medskip

	\begin{subfigure}{0.5\textwidth}
		\includegraphics[width=\linewidth]{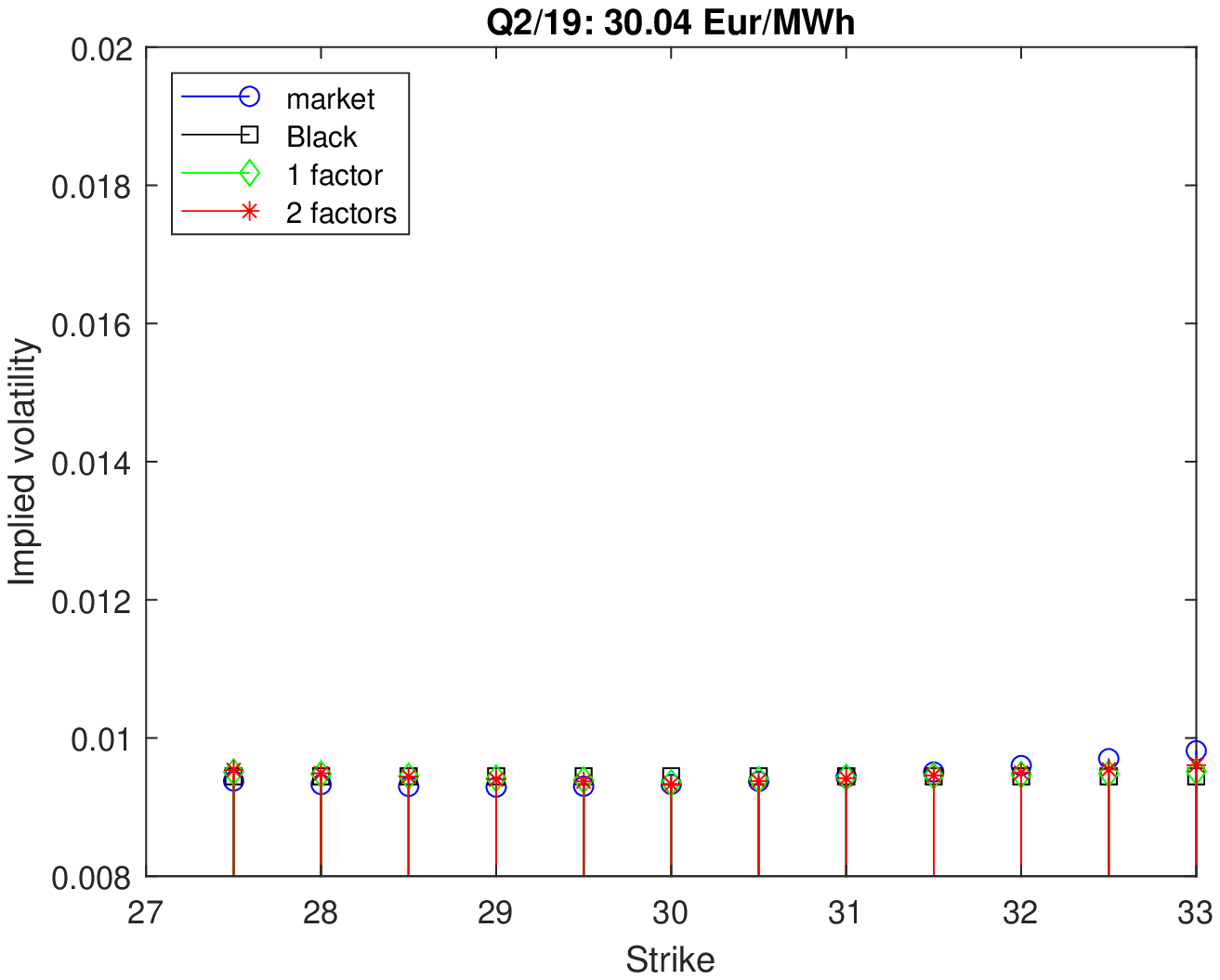}
	\end{subfigure}\hspace*{\fill}
	\begin{subfigure}{0.5\textwidth}
		\includegraphics[width=\linewidth]{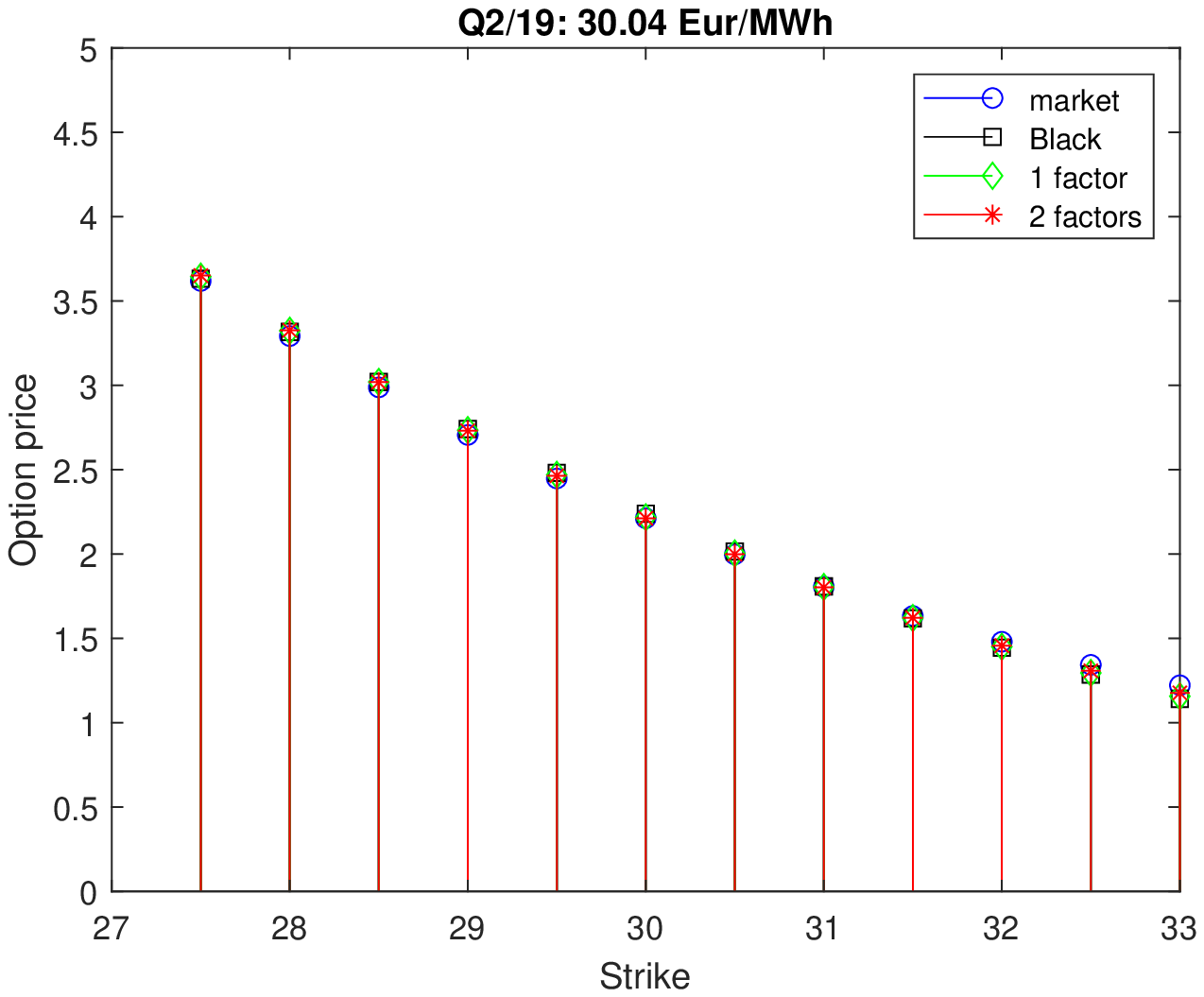}
	\end{subfigure}
	
	\vspace{30pt}
	\medskip

	\begin{subfigure}{0.5\textwidth}
		\includegraphics[width=\linewidth]{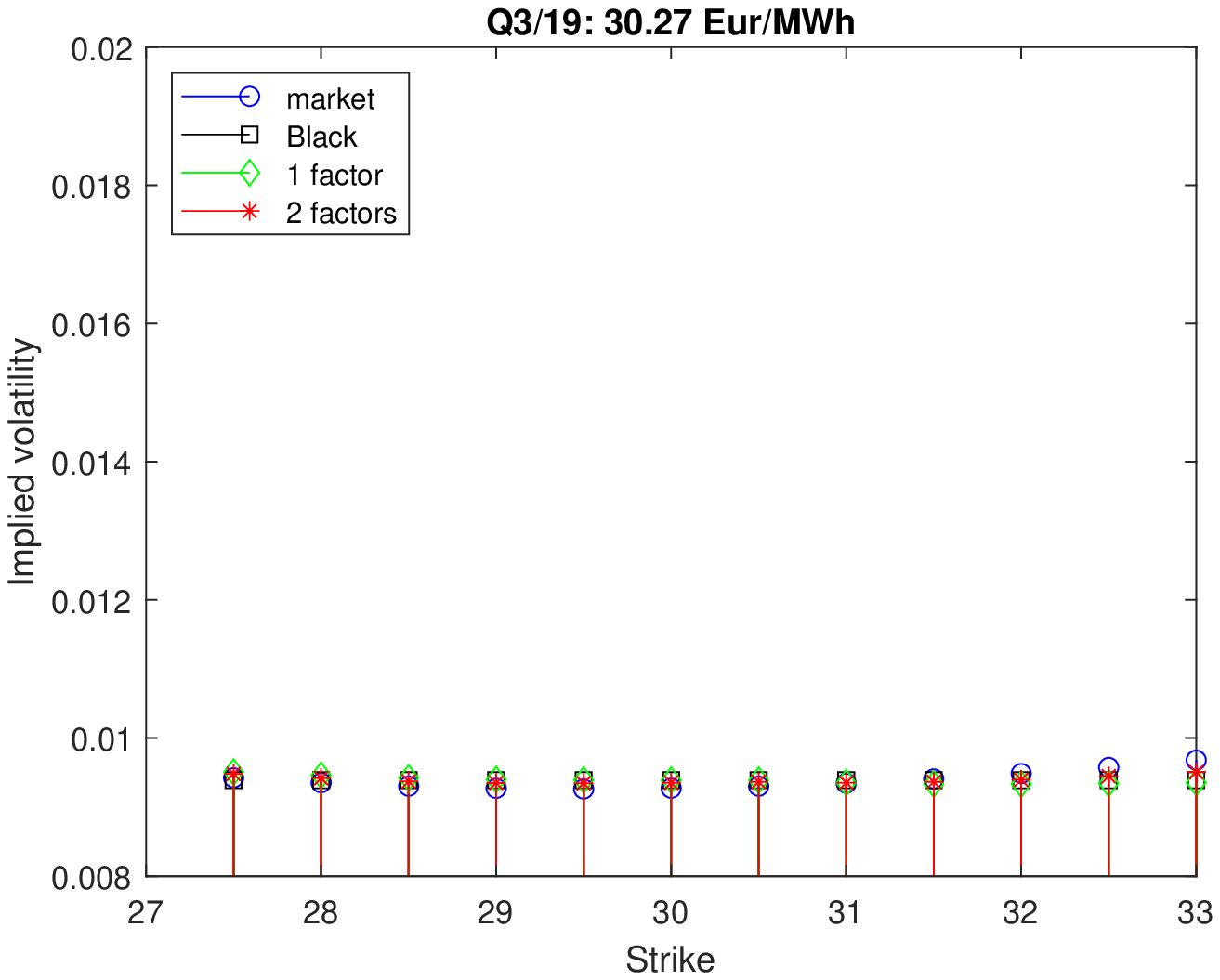}
	\end{subfigure}\hspace*{\fill}
	\begin{subfigure}{0.5\textwidth}
		\includegraphics[width=\linewidth]{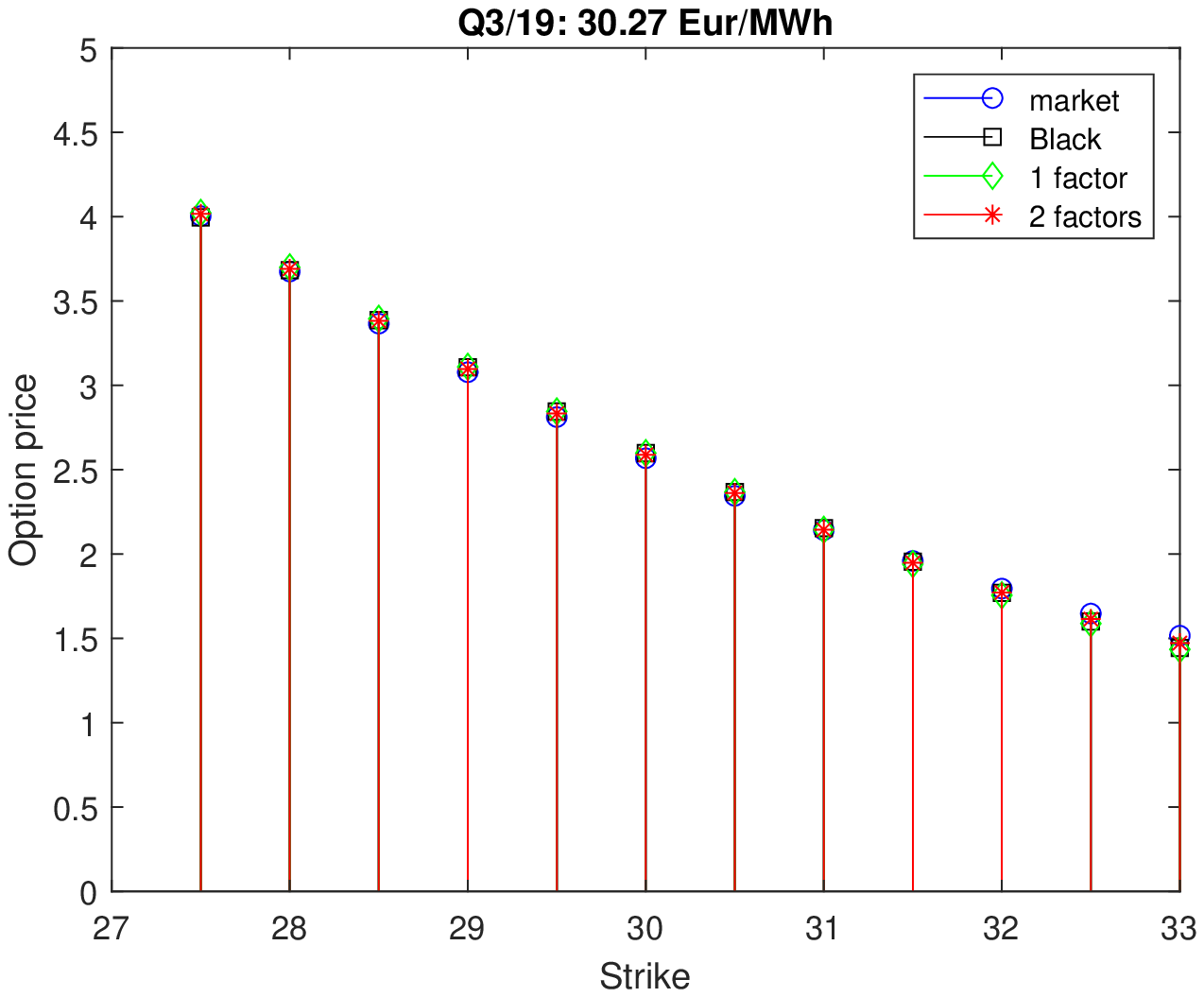}
	\end{subfigure}
	
	
	
	\caption{Implied volatility for the Black, one-factor and two-factor model compared to the empirical implied volatilities of all the options listed at March 5, 2018 (left). The corresponding underlying current price is indicated above each plot. On the right hand side the corresponding prices are shown. Quarterly delivery periods in 2019.} \label{fig:2_ivq2}
\end{figure}

\begin{figure} 
	\begin{subfigure}{0.5\textwidth}
		\includegraphics[width=\linewidth]{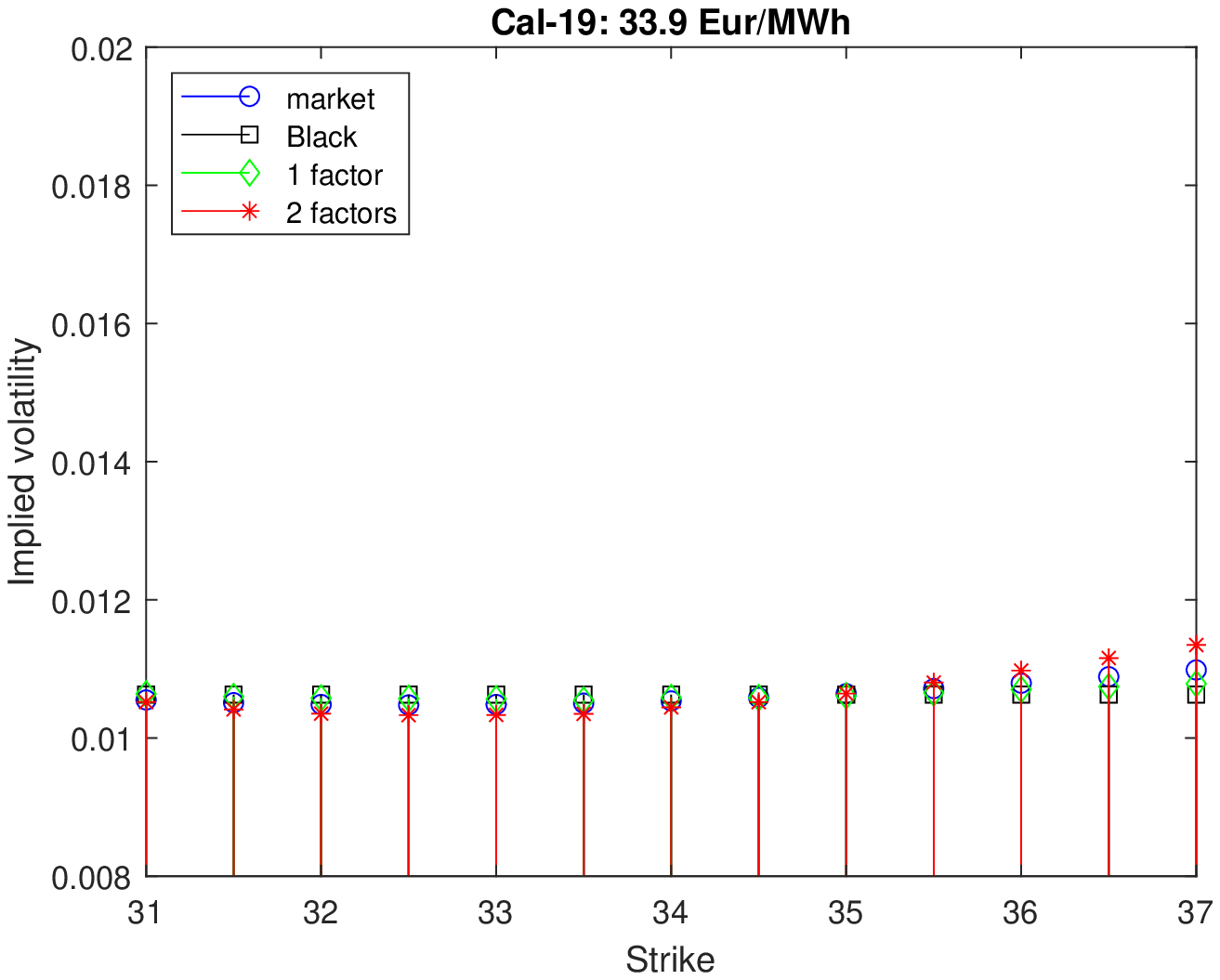}
	\end{subfigure}\hspace*{\fill}
	\begin{subfigure}{0.5\textwidth}
		\includegraphics[width=\linewidth]{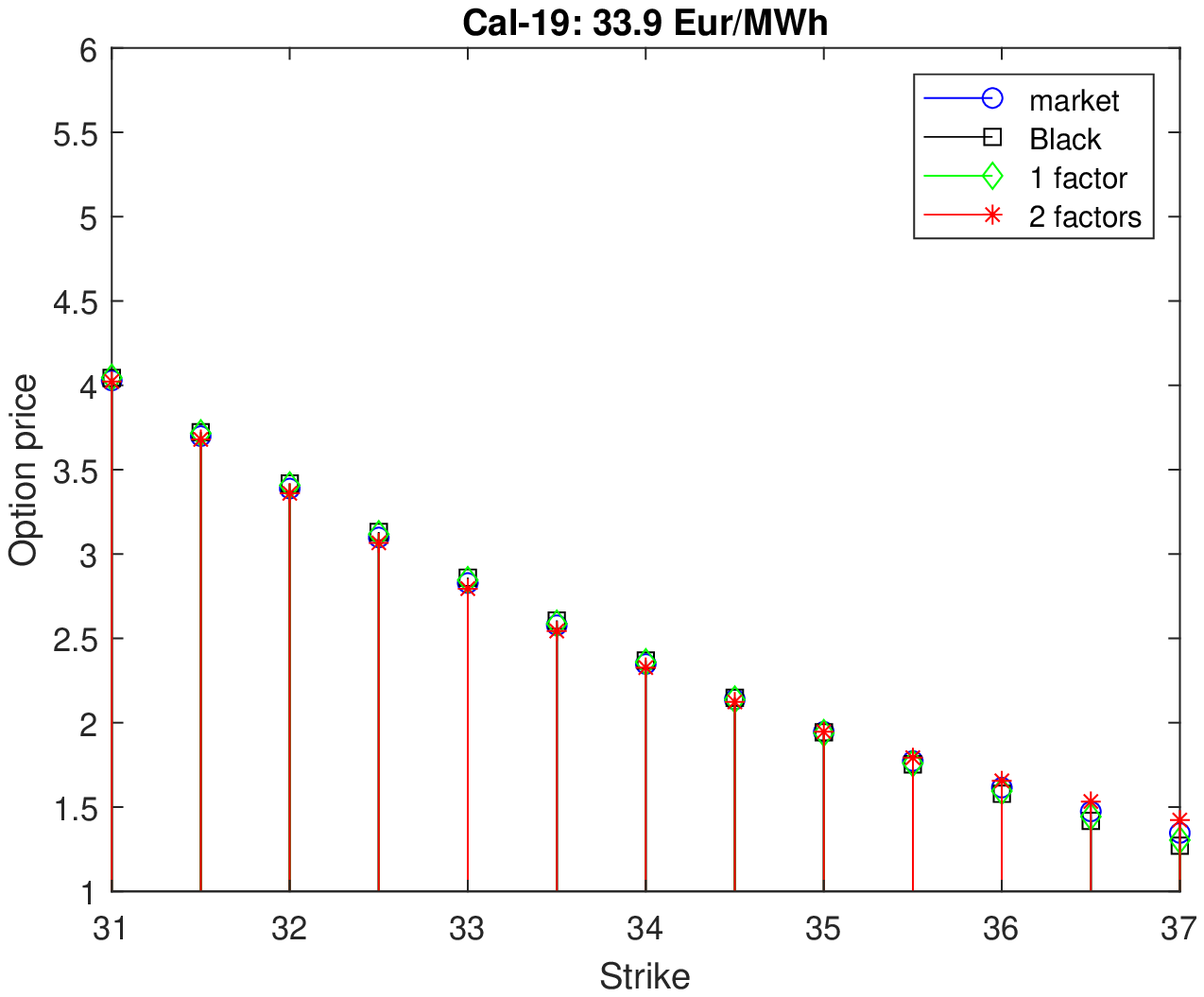}
	\end{subfigure}
	
	\vspace{30pt}
	\medskip
	\begin{subfigure}{0.5\textwidth}
		\includegraphics[width=\linewidth]{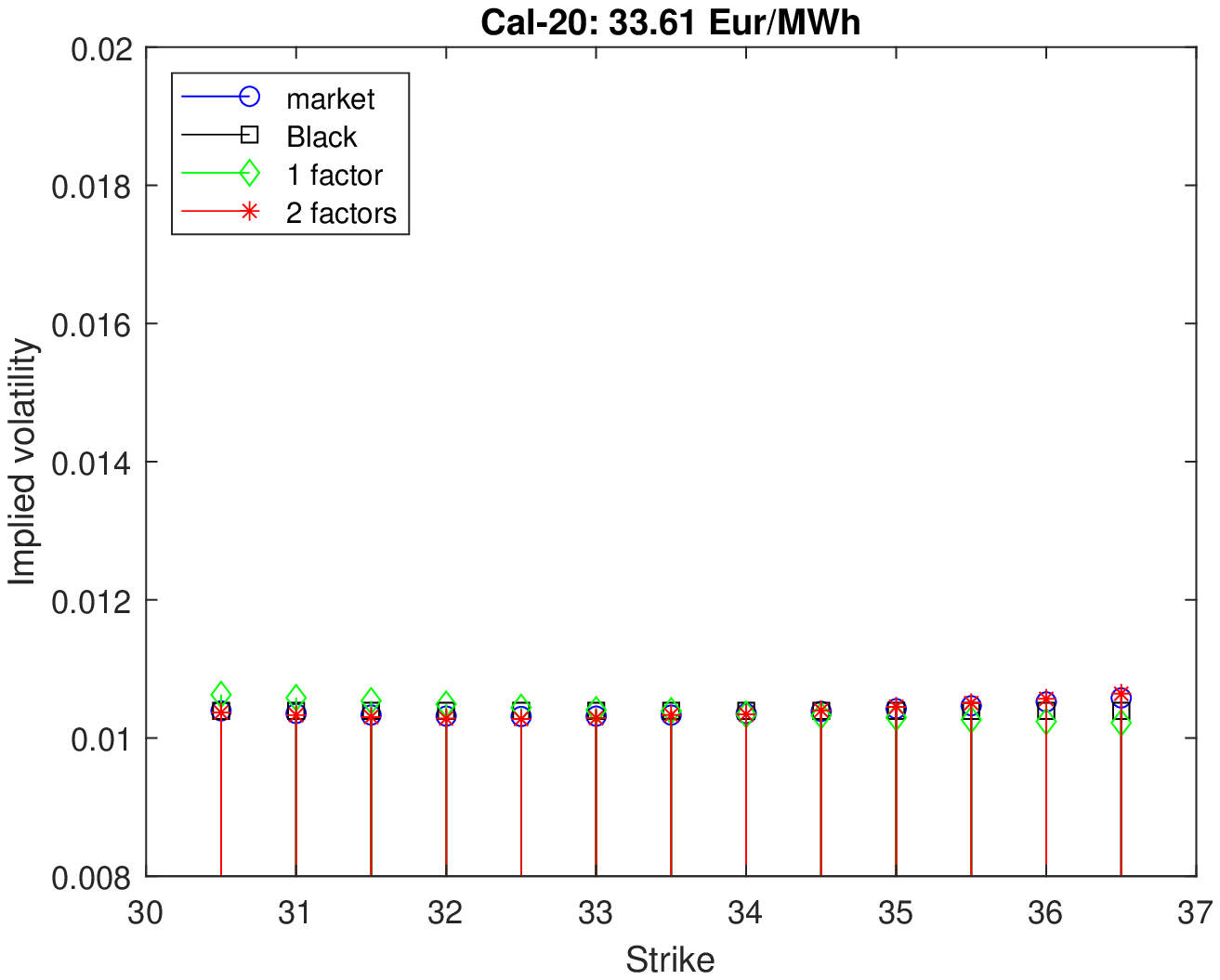}
	\end{subfigure}\hspace*{\fill}
	\begin{subfigure}{0.5\textwidth}
		\includegraphics[width=\linewidth]{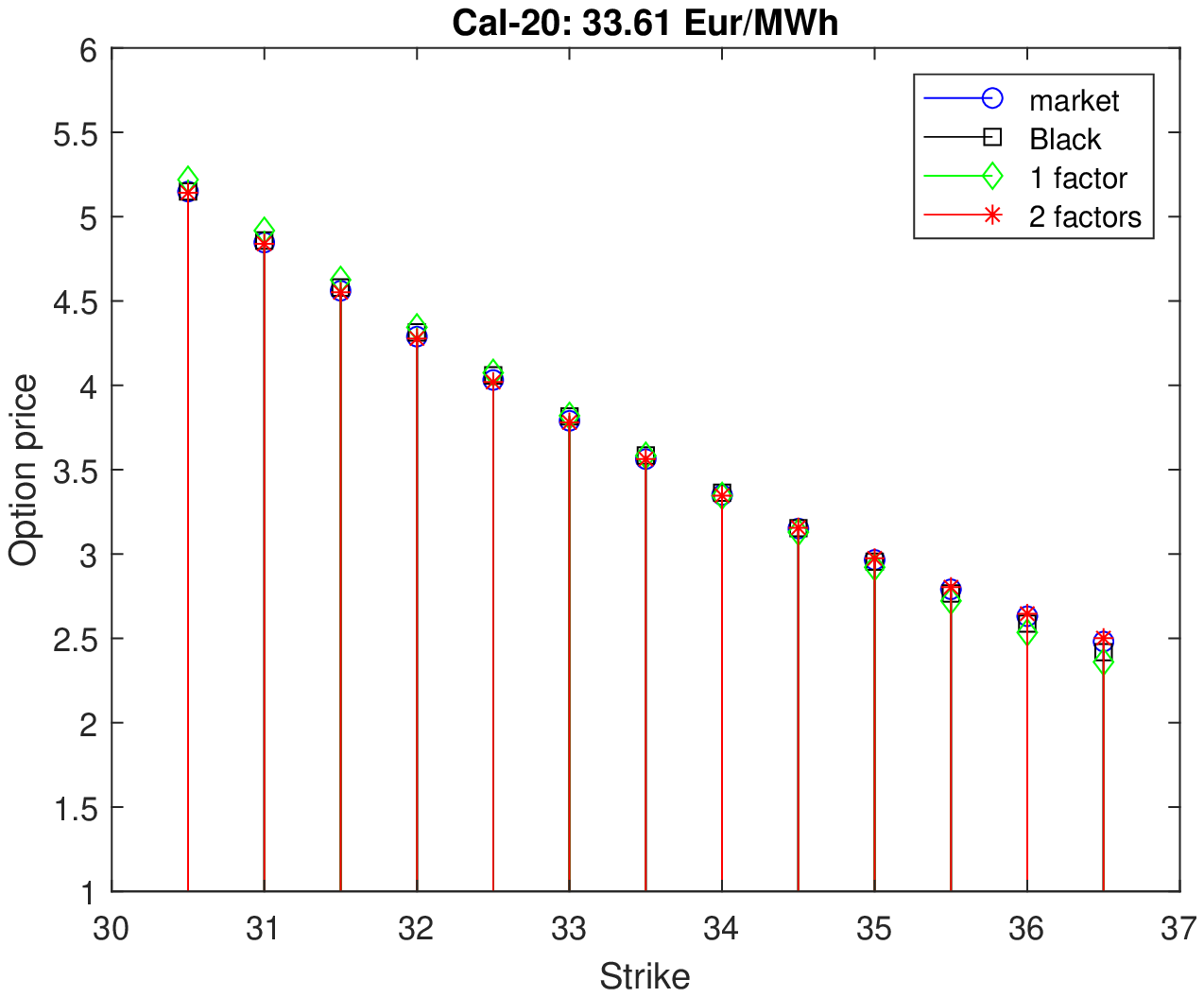}
	\end{subfigure}
	
	\vspace{30pt}
	\medskip
	
	\begin{subfigure}{0.5\textwidth}
		\includegraphics[width=\linewidth]{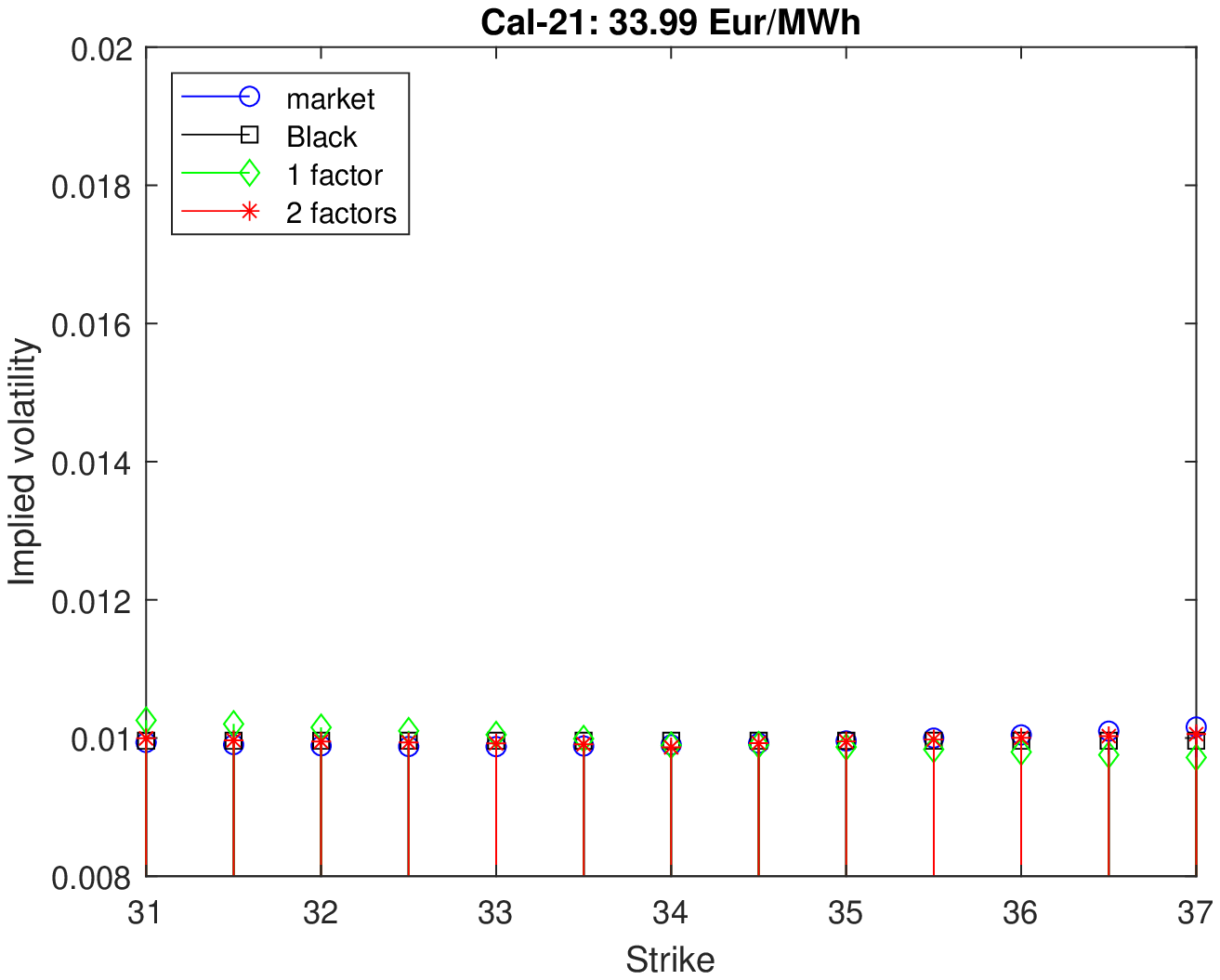}
	\end{subfigure}\hspace*{\fill}
	\begin{subfigure}{0.5\textwidth}
		\includegraphics[width=\linewidth]{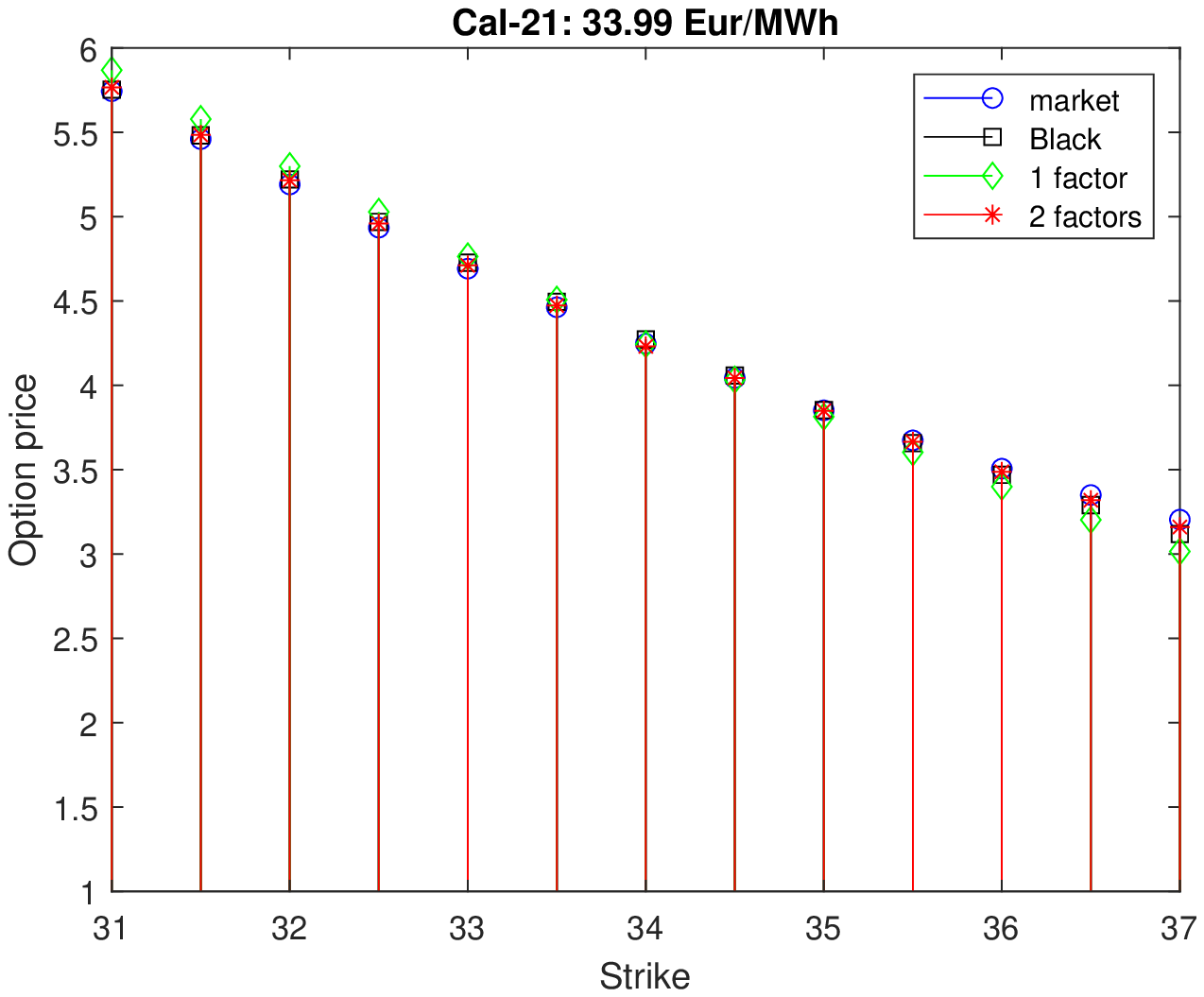}
	\end{subfigure}

	\vspace{30pt}
	
	\caption{Implied volatility for the Black, one-factor and two-factor model compared to the empirical implied volatilities of all the options listed at March 5, 2018 (left). The corresponding underlying current price is indicated above each plot. On the right hand side the corresponding prices are shown. Yearly delivery periods.} \label{fig:3_iv}
\end{figure}

\newpage
\bibliographystyle{plain}

\bibliography{bibbase}

\def\cprime{$'$} \def\cprime{$'$} \def\cprime{$'$}
\begin{thebibliography}{10}

\bibitem{AS}
M.~Abramowitz and I.~A. Stegun.
\newblock {\em Handbook of mathematical functions with formulas, graphs, and
  mathematical tables}, volume~55 of {\em National Bureau of Standards Applied
  Mathematics Series}.
\newblock For sale by the Superintendent of Documents, U.S. Government Printing
  Office, Washington, D.C., 1964.

\bibitem{andresen}
A.~Andresen, S.~Koekebakker, and S.~Westgaard.
\newblock Modeling electricity forward prices using the multivariate normal
  inverse gaussian distribution.
\newblock {\em The Journal of Energy Markets}, 3(3):3, 2010.

\bibitem{arismendi}
J.~C. Arismendi, J.~Back, M.~Prokopczuk, R.~Paschke, and M.~Rudolf.
\newblock {Seasonal Stochastic Volatility: Implications for the pricing of
  commodity options}.
\newblock {\em Journal of Banking \& Finance}, 66(C):53--65, 2016.

\bibitem{back}
J.~Back, M.~Prokopczuk, and M.~Rudolf.
\newblock {Seasonality and the valuation of commodity options}.
\newblock {\em Journal of Banking \& Finance}, 37(2):273--290, 2013.

\bibitem{BN}
O.~E. Barndorff-Nielsen.
\newblock Normal {I}nverse {G}aussian distributions and stochastic volatility
  modelling.
\newblock {\em Scand. J. Statist.}, 24(1):1--13, 1997.

\bibitem{BN97}
O.~E. Barndorff-Nielsen.
\newblock Processes of {N}ormal {I}nverse {G}aussian type.
\newblock {\em Finance and Stochastics}, 2(1):41--68, 1997.

\bibitem{benth04}
F.~E. Benth and J.~Šaltytė Benth.
\newblock The normal inverse gaussian distribution and spot price modelling in
  energy markets.
\newblock {\em International Journal of Theoretical and Applied Finance},
  07(02):177--192, 2004.

\bibitem{MR2323278}
F.~E. Benth, J.~Kallsen, and T.~Meyer-Brandis.
\newblock A non-{G}aussian {O}rnstein-{U}hlenbeck process for electricity spot
  price modeling and derivatives pricing.
\newblock {\em Appl. Math. Finance}, 14(2):153--169, 2007.

\bibitem{kufa}
F.~E. Benth and R.~Kufakunesu.
\newblock Pricing of exotic energy derivatives based on arithmetic spot models.
\newblock {\em International Journal of Theoretical and Applied Finance},
  12(04):491--506, 2009.

\bibitem{BPV}
F.~E. Benth, M.~Piccirilli, and T.~Vargiolu.
\newblock Mean-reverting additive energy forward curves in a
  {Heath-Jarrow-Morton} framework.
\newblock {\em Mathematics and Financial Economics. In press.}, 2019.

\bibitem{schmeck76}
F.~E. Benth and M.~D. Schmeck.
\newblock Pricing and hedging options in energy markets using {Black-76}.
\newblock {\em Journal of Energy Markets}, 7(2):35--69, 2014.

\bibitem{schmeck_rp}
F.~E. Benth and M.~D. Schmeck.
\newblock {\em Pricing Futures and Options in Electricity Markets}, pages
  233--260.
\newblock Springer Berlin Heidelberg, Berlin, Heidelberg, 2014.

\bibitem{MR2416686}
F.~E. Benth, J.~{{\v{S}}altyt{\.e} Benth}, and S.~Koekebakker.
\newblock {\em {Stochastic Modelling of Electricity and Related Markets}}.
\newblock Advanced Series on Statistical Science \& Applied Probability, 11.
  World Scientific Publishing Co. Pte. Ltd., Hackensack, NJ, 2008.

\bibitem{MR2057475}
N.~H. Bingham and R.~Kiesel.
\newblock {\em Risk-neutral valuation}.
\newblock Springer Finance. Springer-Verlag London, Ltd., London, second
  edition, 2004.
\newblock Pricing and hedging of financial derivatives.

\bibitem{birkelund}
O.~H. Birkelund, E.~Haugom, P.~Molnár, M.~Opdal, and S.~Westgaard.
\newblock A comparison of implied and realized volatility in the {N}ordic power
  forward market.
\newblock {\em Energy Economics}, 48:288 -- 294, 2015.

\bibitem{black}
F.~Black.
\newblock The pricing of commodity contracts.
\newblock {\em Journal of Financial Economics}, 3(1):167 -- 179, 1976.

\bibitem{borger}
R.~B\"orger.
\newblock Energy-related commodity futures: statistics, models and derivatives.
\newblock 2007.
\newblock {Ph.D. Dissertation}.

\bibitem{borovkova2017electricity}
S.~Borovkova and M.~D. Schmeck.
\newblock Electricity price modeling with stochastic time change.
\newblock {\em Energy Economics}, 63:51--65, 2017.

\bibitem{brigo_mercurio}
D.~Brigo and F.~Mercurio.
\newblock {\em Interest rate models---theory and practice}.
\newblock Springer Finance. Springer-Verlag, Berlin, second edition, 2006.
\newblock With smile, inflation and credit.

\bibitem{CM}
P.~Carr and D.~B. Madan.
\newblock Option valuation using the fast fourier transform.
\newblock {\em Journal of Computational Finance}, 2:61--73, 1999.

\bibitem{MR2042661}
R.~Cont and P.~Tankov.
\newblock {\em {Financial Modelling with Jump Processes}}.
\newblock Chapman \& Hall/CRC Financial Mathematics Series. Chapman \&
  Hall/CRC, Boca Raton, FL, 2004.

\bibitem{CT}
R.~Cont and P.~Tankov.
\newblock Nonparametric calibration of jump-diffusion option pricing models.
\newblock {\em Journal of Computational Finance}, 7:1--49, 2004.

\bibitem{gallana}
E.~Edoli, M.~Gallana, and T.~Vargiolu.
\newblock Optimal intra-day power trading with a {G}aussian additive process.
\newblock {\em Journal of Energy Markets}, 10(4):23--42, 2017.

\bibitem{eex2}
EEX.
\newblock {Trading Volume in Power Derivatives, June 2018}.
\newblock
  \url{https://www.eex.com/en/about/newsroom/news-detail/eex-trading-results-for-june-2018/83548}.

\bibitem{musti}
V.~Fanelli, L.~Maddalena, and S.~Musti.
\newblock Modelling electricity futures prices using seasonal path-dependent
  volatility.
\newblock {\em Applied Energy}, 173:92 -- 102, 2016.

\bibitem{fanelli}
V.~Fanelli and M.~D. Schmeck.
\newblock On the seasonality in the implied volatility of electricity options.
\newblock {\em Quantitative Finance}, 19(8):1321--1337, 2019.

\bibitem{lautier}
E.~Jaeck and D.~Lautier.
\newblock Volatility in electricity derivative markets: The {S}amuelson effect
  revisited.
\newblock {\em Energy Economics}, 59:300 -- 313, 2016.

\bibitem{kieselpara}
R.~Kiesel and F.~Paraschiv.
\newblock Econometric analysis of 15-minute intraday electricity prices.
\newblock {\em Energy Economics}, 64:77 -- 90, 2017.

\bibitem{kieselSB}
R.~Kiesel, G.~Schindlmayr, and R.~H. Börger.
\newblock A two-factor model for the electricity forward market.
\newblock {\em Quantitative Finance}, 9(3):279--287, 2009.

\bibitem{latini}
L.~Latini, M.~Piccirilli, and T.~Vargiolu.
\newblock Mean-reverting no-arbitrage additive models for forward curves in
  energy markets.
\newblock {\em Energy Economics}, 79:157--170, 2019.

\bibitem{nastasi}
E.~{Nastasi}, A.~{Pallavicini}, and G.~{Sartorelli}.
\newblock Smile modelling in commodity markets.
\newblock Preprint, 2018.

\bibitem{nomikos}
N.~K. Nomikos and O.~A. Soldatos.
\newblock {Analysis of model implied volatility for jump diffusion models:
  Empirical evidence from the Nordpool market}.
\newblock {\em Energy Economics}, 32(2):302--312, 2010.

\bibitem{schmeck2016pricing}
M.~D. Schmeck.
\newblock Pricing options on forwards in energy markets: the role of mean
  reversion's speed.
\newblock {\em International Journal of Theoretical and Applied Finance},
  19(8):1650053, 2016.

\bibitem{schneider16}
L.~Schneider and B.~Tavin.
\newblock From the samuelson volatility effect to a samuelson correlation
  effect: An analysis of crude oil calendar spread options.
\newblock {\em Journal of Banking \& Finance}, 2016.

\bibitem{schneider18}
L.~{Schneider} and B.~{Tavin}.
\newblock {The Samuelson Effect and Seasonal Stochastic Volatility in
  Agricultural Futures Markets}.
\newblock Preprint, 2018.

\bibitem{taylor11}
S.~J. Taylor.
\newblock {\em Asset Price Dynamics, Volatility, and Prediction}.
\newblock {Princeton University Press}, 2011.

\bibitem{wottka}
T.~Wottka.
\newblock Volatility skews implied by a multi-technology bid stack model for
  electricity markets.
\newblock Preprint, 2017.

\end{thebibliography}

\end{document}